\documentclass[9pt, preprint]{sigplanconf}

\usepackage{amsmath,proof}

\usepackage{etex}
\usepackage{amssymb}
\usepackage{amsmath}
\usepackage{amsthm}
\usepackage{mathtools}
\usepackage{stmaryrd}
\usepackage{color}
\usepackage{comment}

\theoremstyle{theorem}
\newtheorem{theorem}{Theorem}[section]
\newtheorem{proposition}[theorem]{Proposition}
\newtheorem{lemma}[theorem]{Lemma}
\newtheorem{corollary}[theorem]{Corollary}
\theoremstyle{definition}
\newtheorem{definition}[theorem]{Definition}
\newtheorem{example}[theorem]{Example}

\newtheorem{remark}[theorem]{Remark}

\newtheorem{notation}[theorem]{Notation}

\newcommand{\memo}[1]{\textcolor{red}{(#1)}}
\newcommand{\cont}{\memo{cont.}}

\newcommand{\notyet}{\textcolor{red}{(not yet)\ }}

\newcommand{\bbR}{\mathbb{R}}

\newcommand{\calF}{\mathcal{F}}
\newcommand{\calL}{\mathcal{L}}

\newcommand{\calS}{\mathcal{S}}

\newcommand{\cat}[1]{\mathcal{#1}}

\renewcommand{\bar}{\overbar}

\newcommand{\overbar}[1]{\mkern 1.5mu\overline{\mkern-1.5mu#1\mkern-1.5mu}\mkern 1.5mu}

\renewcommand{\phi}{\varphi}

\newcommand{\xxtos}[1]{\stackrel{#1}{\to}}
\newcommand{\xxto}{\xrightarrow}

\newcommand{\tto}{\rightrightarrows}

\newcommand{\into}{\rightarrowtail}
\newcommand{\onto}{\twoheadrightarrow}

\newcommand{\pow}{\mathcal{P}}
\newcommand{\dist}{\mathcal{D}_{=1}}
\newcommand{\sdist}{\mathcal{D}_{\le 1}}

\newcommand{\set}[2]{\left\{\, #1 \mathrel{}\middle|\mathrel{} #2 \,\right\}}
\newcommand{\sett}[1]{\left\{ #1 \right\}}

\newcommand{\place}{{-}}

\newcommand{\op}{\mathrm{op}}
\newcommand{\ev}{\mathrm{ev}}
\newcommand{\id}{\mathrm{id}}

\newcommand{\uintv}{[0, 1]}

\newcommand{\Set}{\mathbf{Set}}
\newcommand{\Sets}{\Set}
\newcommand{\Pos}{\mathbf{Pos}}
\newcommand{\Rel}{\mathbf{Rel}}
\newcommand{\Top}{\mathbf{Top}}
\newcommand{\CL}{\mathbf{CL}}
\newcommand{\GEMod}{\mathbf{GEMod}}
\newcommand{\DcEMod}{\mathbf{DcEMod}}
\newcommand{\Conv}{\mathbf{Conv}}

\newcommand{\EM}{\mathcal{E}{\kern-.5ex}\mathcal{M}}
\newcommand{\Kl}{\mathcal{K}{\kern-.2ex}\ell}

\newcommand{\biglor}{\bigvee}
\newcommand{\bigland}{\bigwedge}

\DeclarePairedDelimiter\bra{\langle}{\rvert}
\DeclarePairedDelimiter\ket{\lvert}{\rangle}
\DeclarePairedDelimiter\abs{\lvert}{\rvert}
\DeclarePairedDelimiter\p{(}{)}

\newcommand{\OmegaD}{\Omega_{\cat{D}}}

\newcommand{\List}{\mathsf{List}}

\newcommand{\Upx}{\mathsf{Up}}

\newcommand{\lt}{<}
\newcommand{\gt}{>}

\renewcommand{\subset}{\subseteq}

\newcommand{\StrCL}{\mathrm{\textbf{StrCL}}_{\bigland,\mathrm{pos}}}

\newif\ifignore 
\ignorefalse
\newcommand{\auxproof}[1]{
\ifignore\mbox{}\newline
\textbf{BEGIN: AUX-PROOF} \dotfill\newline
{#1}\mbox{}\newline
\textbf{END: AUX-PROOF}\dotfill\newline
\fi}

\usepackage[all,dvips]{xy}

\xyoption{v2}
\xyoption{curve}
\xyoption{2cell}
\SelectTips{cm}{}  
\UseAllTwocells
\SilentMatrices

\newdir{ >}{{}*!/-8pt/@{>}}
\newdir{|>}{%
!/1.6pt/@{|}*:(1,-.2)@^{>}*:(1,+.2)@_{>}}
\newdir{pb}{:(1,-1)@^{|-}}
\def\pb#1{\save[]+<20 pt,0 pt>:a(#1)\ar@{pb{}}[]\restore}

\usepackage{wrapfig}
\def\myqed{\qed}

\newcommand{\after}{\mathrel{\circ}}
\newcommand{\co}{\mathrel{\circ}}

\newcommand{\Kleisli}[1]{\mathcal{K}{\kern-.2ex}\ell(#1)}

\newcommand{\bbP}{\mathbb{P}}

\newcommand{\fpow}{\mathcal{P}_{\omega}}
\newcommand{\lift}{\mathcal{L}}
\newcommand{\iso}{\mathrel{\stackrel{
\raisebox{.5ex}{$\scriptstyle\cong\,$}}{
\raisebox{0ex}[0ex][0ex]{$\rightarrow$}}}}

\newcommand{\To}{\Rightarrow}
\newcommand{\longto}{\longrightarrow}

\def\compsign{\mathrel>\kern-2pt\joinrel>\kern-2pt\joinrel>}

\newcommand{\valg}[3]{\raisebox{.00in}
{\mbox{\large
${{{\scriptstyle #1}\atop
{\phantom{\scriptstyle #2}}\scriptstyle\downarrow #2}
\atop{\scriptstyle #3}}$}}}

\newcommand{\dar}{\ar@{..>}}
\newcommand{\lar}{\ar@{-}}

\newcommand{\ttrue}{\mathtt{t{\kern-1.5pt}t}}
\newcommand{\ffalse}{\mathtt{f{\kern-1.5pt}f}}

\newcommand{\ssub}{\mathrel{\subset{\kern-1.6ex}\subset}}

\newcommand{\wpre}{\mathop{\mathrm{wp}}\nolimits}

\newcommand{\EMod}{\mathbf{EMod}}

\newcommand{\tauTotal}{\tau_{\mathrm{total}}}
\newcommand{\tauPartial}{\tau_{\mathrm{partial}}}

\newcommand{\PredKl}[1]{\bbP^{\mathcal{K}{\kern-.2ex}\ell}(#1)}
\newcommand{\PredEM}[1]{\bbP^{\mathcal{E}{\kern-.5ex}\mathcal{M}}(#1)}

\newcommand{\upcl}{\mathop{\uparrow}\nolimits}

\newcommand{\UP}{\mathcal{U{\kern-.3ex}P}}
\newcommand{\CD}{\mathcal{C{\kern-.3ex}D}}

\newcommand{\V}{\mathcal{V}}

\newcommand{\RC}{\mathsf{Cv}}

\begin{document}

\setlength{\pdfpageheight}{\paperheight}
\setlength{\pdfpagewidth}{\paperwidth}

\conferenceinfo{CONF 'yy}{Month d--d, 20yy, City, ST, Country}
\copyrightyear{20yy}
\copyrightdata{978-1-nnnn-nnnn-n/yy/mm}
\doi{nnnnnnn.nnnnnnn}




\titlebanner{}        
\preprintfooter{}   

\title{Healthiness from Duality}

\authorinfo{Wataru Hino \and Hiroki Kobayashi \and\\ Ichiro Hasuo}
           {University of Tokyo, Japan}
           {\{wataru, hkoba7de, ichiro\}@is.s.u-tokyo.ac.jp}
\authorinfo{Bart Jacobs}
           {Radboud University Nijmegen, the Netherlands}
           {bart@cs.ru.nl}

\maketitle

\begin{abstract}
  \emph{Healthiness} is a good old question in program logics that
  dates back to Dijkstra. It asks for an intrinsic characterization of
  those predicate transformers which arise as the (backward)
  interpretation of a certain class of programs. There are several
  results known for healthiness conditions: for deterministic
  programs, nondeterministic ones, probabilistic ones, etc.  Building
  upon our previous works on so-called \emph{state-and-effect
    triangles}, we contribute a unified categorical framework for
  investigating healthiness conditions. This framework is based on a
  \emph{dual adjunction} induced by a dualizing object and on our
  notion of \emph{relative Eilenberg-Moore algebra}.  The latter
  notion seems interesting in its own right in the context of monads,
  Lawvere theories and enriched categories.
\end{abstract}





\section{Introduction}\label{sec:intro}

\paragraph{Predicate Transformer Semantics of Computation}
\emph{Program logics} are formal systems for reasoning about
programs. They come in different styles:
in the \emph{Floyd-Hoare logic}~\cite{Hoare69} one derives 
 triples
of a
precondition, a program and a postcondition;  \emph{dynamic
logics}~\cite{HarelTK00}
are  logics that have programs as modal operators;
type-theoretic presentations would have predicates as \emph{refinement}
(or \emph{dependent})
\emph{types}, allowing smooth extension to higher-order programs; and many
 program verification tools for imperative programs have
programs represented as \emph{control flow graphs}, where predicates
are labels to the edges. Whatever presentation style is taken, the
basic idea that underlies these variations of program logics is that of
\emph{weakest precondition}, dating back to
Dijkstra~\cite{Dijkstra76}. It asks: \emph{in order to guarantee a given
postcondition after the execution of a given program, what
precondition does it suffice to assume, before the execution?}

Through weakest preconditions a program gives rise to a \emph{(backward)
predicate transformer} that carries a given postcondition to the
corresponding weakest precondition. This way of interpreting
programs---sometimes called \emph{axiomatics
semantics}~\cite{Winskel93}---is in contrast to \emph{(forward) state
transformer semantics} where programs are understood as functions
(possibly with branching or side effects) that carry input states/values
to output ones.

\paragraph{Predicate Transformer Semantics and Quantum Mechanics}
The topic of weakest precondition and predicate transformer semantics
is  classic in computer science, in decades of foundational and
practical studies. Recently, fresh light has been shed on their
\emph{structural} aspects: the same kind of interplay between
\emph{dynamics} and \emph{observations} for \emph{quantum mechanics}
and \emph{quantum logic} appears in predicate transformer semantics,
as noted by one of the current authors---together with his
colleagues~\cite{Jacobs14CMCS,jacobs2015dijkstra,Jacobs15LMCS}.  This
enabled them to single out a simple categorical scheme---called
\emph{state-and-effect triangles}---that is shared by program
semantics and quantum mechanics.

On the program semantics side,
the scheme of state-and-effect triangles
allows
the informal ``duality'' between  state and predicate transformer
semantics to be formalized as a
categorical duality. Interestingly, the quantum counterpart of this
duality is the one between the \emph{Schr\"{o}dinger} and
\emph{Heisenberg} pictures of quantum mechanics. In this sense the idea
of weakest precondition dates back before Dijkstra, and before the
notion of program.

 State-and-effect triangles will be
elaborated on in Section~\ref{sub:stateAndEffectTriangles}; we note at this
stage that the term ``effect'' in the name refers to a notion in quantum
mechanics and should be read as \emph{predicate} in the programming
context. In particular, it has little to do with \emph{computational
effect}.

\paragraph{In Search of Healthiness}
The question of \emph{healthiness conditions} is one that is as old as
the idea of weakest precondition~\cite{Dijkstra76}: it asks for an intrinsic characterization of those predicate
  transformers which arise as the (backward) interpretation of
  programs.
One basic healthiness result is for \emph{nondeterministic}
programs.  The result is stated, in elementary terms, as follows.
\begin{theorem}[healthiness under the ``may''-nondeterminism]
\label{thm:healthiness-nondet-elementary}
 \begin{enumerate}
  \item Let $R\subseteq X\times Y$ be a binary relation; it is thought
	of as a nondeterministic computation from  $X$ to $Y$. This
	$R$ induces a predicate transformer ($\wpre$ for
	``weakest precondition'')
	\begin{align*}
	 &	\wpre_{\Diamond}(R)\colon 2^{Y}\longrightarrow 2^{X},
	 \quad\text{defined by}\quad
	 \\
	 &
	 \wpre_{\Diamond}(R)(f)(x)=1
	 \;\Longleftrightarrow\;
	 \exists y\in Y.\, (xRy\;\land\;f(y)=1),
	\end{align*}
	for each $f\colon Y\to 2$ (thought of as a \emph{predicate} and
	more specifically as a \emph{postcondition}) and each $x\in X$.
  \item (Healthiness) Let $\varphi\colon 2^{Y}\to 2^{X}$ be a function. The following
	are equivalent.
	\begin{enumerate}
	 \item The function $\varphi$ arises  in
	       the way prescribed above. That is, there exists
	       $R\subseteq X\times Y$ such that
	       $\varphi=\wpre_{\Diamond}(R)$.
	 \item The map $\varphi$ is \emph{join-preserving}, where $2^{Y}$ and
	       $2^{X}$
	       are equipped with (the pointwise extensions of) the order
	       $0<1$ in $2$.
	\end{enumerate}
\end{enumerate}
\end{theorem}
\noindent
Here we
interpret $0\in 2$ as false and $1\in
2$ as true, a convention we adopt throughout the paper.

There are many different instances of healthiness results.  For example,
the works~\cite{Kozen81,Jones90PhD} study \emph{probabilistic}
computations in place of nondeterministic ones; the (alternating)
combination of nondeterministic and probabilistic branching is studied
in~\cite{MorganMS96}; and Dijkstra's original work~\cite{Dijkstra76}
deals with the (alternating) combination of nondeterminism and
divergence. In fact it is implicit in our notation $\wpre_{\Diamond}$
that there is a possible ``must'' variant of
Theorem~\ref{thm:healthiness-nondet-elementary}. In this
variant, another  predicate transformer $\wpre_{\Box}$ is defined by
\begin{equation}\label{eq:boxPredTransfConcretely}
\small\begin{array}{l}
  	 \wpre_{\Box}(R)(f)(x)=1
	 \quad\Longleftrightarrow\quad
	 \forall y\in Y.\, (xRy\;\Rightarrow\;f(y)=1),
\end{array}
\end{equation}
requiring that every possible poststate must satisfy the postcondition
$f$.  The corresponding healthiness result has it that the resulting predicate
transformers are characterized by \emph{meet-preservation}.

The goal of the current work is to identify a structural and categorical
principle behind  healthiness, and hence to provide a common
ground for the  existing body of healthiness results, also providing a methodology that possibly aids finding
new  results.

 As a concrete instance of this goal, we wish to answer why
join-preservation should characterize ``may''-nondeterministic predicate
transformers $\wpre_{\Diamond}$ in
Theorem~\ref{thm:healthiness-nondet-elementary}. A first observation
would be that the powerset monad $\pow$---that occurs in the alternative
description $R\colon X\to \pow Y$ of a binary relation $R$---has
complete join-semilattices as its Eilenberg-Moore algebras.  This alone
should not be enough though---the framework needs to account for
different modalities, such as $\Diamond$ (``may'') vs.\ $\Box$
(``must'') for nondeterminism. (In fact it turns out that this ``first
observation'' is merely a coincidence. See
Section~\ref{sub:diamond-modality} later.)

\paragraph{Our Contributions}
We shall answer to the above question of ``categorical healthiness
condition'' by unifying two constructions---or \emph{recipes}---of
state-and-effect triangles.
\begin{itemize}
 \item One recipe~\cite{Hasuo14,Hasuo15TCS} is called the \emph{modality} one,
       whose modeling of situations like in
       Theorem~\ref{thm:healthiness-nondet-elementary} is centered
       around the notion of \emph{monad}. Firstly,
       the relevant class of
       computations (nondeterministic, diverging,
       probabilistic, etc.) is determined by a monad $T$, and
       a computation is then a function of the type $X\to TY$. Secondly,
              the set $\Omega$ of truth values (such as $2$
       in Theorem~\ref{thm:healthiness-nondet-elementary})
       carries a $T$-algebra
       $\tau\colon T\Omega\to\Omega$; it represents a modality such
       as $\Diamond$ and $\Box$.
 \item The other recipe~\cite{Jacobs15CALCO} is referred to as the
       \emph{dual adjunction} one. It takes a dual adjunction
       \begin{math}
		\xymatrix@1@C-.5em{
	 {\cat{C}}
	   \ar@/^1ex/[r]
	   \ar@{}[r]|-{\bot}
        &
	 {\cat{D}\rlap{$^{\op}$}
	}
	   \ar@/^1ex/[l]
	}
       \end{math}\quad
       as an ingredient; and uses two \emph{comparison functors}---from
       a Kleisli category and to an Eilenberg-Moore category---to form a
       state-and-effect triangle, additionally exploiting $\cat{D}$'s
       completeness assumption. One notable feature is that the
       resulting state-and-effect triangle is automatically
       ``healthy''---this is because comparison functors are full and
       faithful.
\end{itemize}
Combining the two recipes we take advantages of both: the former
provides a concrete presentation of predicate transformers by a
modality; and the latter establishes healthiness.  We demonstrate that
many known healthiness results are instances of this framework.

The key to combining the two recipes is to interpret a
monad $T$ on $\Sets$ in a category $\cat{D}$ that is other than
$\Sets$. For this purpose---assuming that the dual adjunction in the
second recipe is given with a dualizing
object---we introduce the notion of \emph{$\cat{D}$-relative $T$-algebra}
and develop its basic theory. Notably the structure map of a $\cat{D}$-relative $T$-algebra  is given by a \emph{monad map} from $T$ to a suitable
continuation-like monad (that arises from the aforementioned dual
adjunction). This notion seems to be more than a tiny side-product of the
current venture: we expect it to play an important role in the
\emph{categorical model theory} (see e.g.~\cite{AdamekR94,LackP09,MakkaiP89}) where the
equivalence between (finitary) monads and \emph{Lawvere theories} is
fundamental. See below for further discussions.

\paragraph{Related and Future Work}
We believe the current results allow rather straightforward
generalization (from ordinary, $\Sets$-based category theory) to
enriched category theory~\cite{Kelly82}. For example, the use of
the $|X|$-fold product $\Omega^X$ can be replaced by the \emph{cotensor}
$[X,\Omega]$. Doing so, and identification of this generalization's
relevance in program logics, is left as future work.

The current theoretical developments are heavily influenced by
\emph{Lawvere theories}, another categorical formalization of algebraic
structures that is (if finitary) equivalent to monads.
In particular, our notion of relative algebra is aimed to be a (partial)
answer to the oft-heard question: \emph{A Lawvere theory can be interpreted
in different categories. Why not a monad?} We intend to
establish
formal relationships in  future work, possibly in an enriched
setting.
 There the line of works on enriched Lawvere theories will be
relevant~\cite{LackP09,hyland2007category}.
The first observation in this direction is that: a monad $T$ on $\Sets$
gives rise to a (possibly large) ``Lawvere theory'' $\Kl(T)^{\op}$; and
then its ``algebra'' in a category $\mathcal{D}$ (with enough products)
is a product-preserving functor $\Kl(T)^{\op}\to \mathcal{D}$. 

What is definitely lacking in the current work (and in our previous work~\cite{Jacobs15CALCO,Hasuo15TCS}) is  syntax for
programs/computations and program logics. In this direction the
work~\cite{GoncharovS13b} presents a generic  set of inference
rules---that is sound and relatively
complete---for a certain class of monadic computations.

We are grateful to a referee who brought our attention to
recent~\cite{HofmannN15}. Motivated by the modal logic question of
equivalences between Kripke frames and modal algebras---possibly
equipped with suitable topological structures---they are led to a
framework that is close to ours.  Their aim is a dual equivalence
between a Kleisli category $\Kl(T)$ and a category of algebras
$\cat{D}$, and our goal of healthiness (i.e.\ a full and faithful
functor $\Kl(T)\to \cat{D}^{\op}$) comes short of such only by failure of
iso-denseness. Some notable differences are as follows. Firstly,
in~\cite{HofmannN15} principal examples of a monad $T$ is for
nondeterminism, so that a Kleisli arrow is a relation, whereas we have
probability and alternation as other leading examples.  Secondly, in
place of relative algebra (that is our novelty), in~\cite{HofmannN15}
they use the notion of algebra that is syntactically presented with
operations. Unifying the results as well as the motivations of the two
papers is an exciting direction of future research. See also
Remark~\ref{rem:CABA}.

Another closely related work~\cite{Keimel15}
studies healthiness from a domain-theoretic point of
view. While it is based on syntactic presentations of algebras
(differently from our monadic presentations), notable similarity
is found in its emphasis on continuation monads. Its domain-theoretic
setting---every construct is $\mathbf{DCpo}$-enriched---will be relevant
when we wish to accommodate recursion in our current results, too.

\paragraph{Organization of the Paper}
We exhibited our leading example in
Theorem~\ref{thm:healthiness-nondet-elementary}. In
Section~\ref{sec:nondet} we describe its proof---in a categorical
language---and this will motivate our general framework. After
recalling the scheme of state-and-effect triangles
in
Section~\ref{sec:general-healthiness}, in
Section~\ref{sec:relativeAlgebraRecipe} we unify  two known recipes for them to present a
new \emph{relative algebra} recipe. The basic theory of relative
algebras is developed there, too. Section~\ref{sec:prob-examples} is devoted
to probabilistic  instances of our framework. Finally in Section~\ref{sec:two-player-setting}
we
further extend the generic framework to accommodate \emph{alternating}
branching that involve two players typically with conflicting interests.

Some missing proofs are found in the appendix.

\paragraph{Preliminaries and Notations}
We assume familiarity with basic category theory, from references
like~\cite{MacLane71,BarrW85}. We list some categories that we will use, mostly for fixing notations:
 the category $\Set$  of sets and functions;
 the category $\Rel$  of sets and binary relations;
 and the categories
 $\CL_{\biglor}$ and
 $\CL_{\bigland}$ of complete join- and meet-semilattices, and
 join- and meet-preserving maps between them,
 respectively.\footnote{Here a \emph{complete join-semilattice} is a
 poset with arbitrary joins $\biglor$. It is well-known that in this case
 arbitrary meets $\bigland$ exist, too; we say ``join-'' to indicate the notion of
homomorphism we are interested in.
 } Given a monad $T$,
 its \emph{Eilenberg-Moore} and \emph{Kleisli} categories are denoted
 by $\EM(T)$ and $\Kl(T)$, respectively. Their definitions are found
 e.g.\ in~\cite{MacLane71,BarrW85}.

Let $S,T$ be monads on $\cat{C}$. The standard notion of
 \emph{monad map} from $S$ to $T$ is defined by
a natural transformation $\alpha\colon S\to T$ that is compatible with
the monad structures. For its explicit requirements see
Appendix~\ref{appendix:monadMap}.


We shall be using various ``hom-like'' entities such as
homsets, exponentials, cotensors and so on; they are denoted by
$\cat{C}(X, Y)$, $Y^X$, $[X, Y]$, etc. For those entities we abuse
the notations $f^*$ and $f_*$ and use them uniformly for the
\emph{precomposition} and \emph{postcomposition} morphisms, such as:
\begin{align*}
\small\begin{array}{l}
 f^{*} = (-) \after f \colon Z^Y \longrightarrow Z^X
 \quad\text{and}\quad
 f_{*} = f \after (-) \colon X^{Z}\longrightarrow Y^{Z}\enspace,
\end{array}
\end{align*}
for $f\colon X\to Y$.
Another generic notation  we will use for those hom-like entities is
 $(\place)^{\sharp}$ for  correspondences like
\begin{displaymath}
  f\colon B\to A^{X}\quad\bigl/\!\!\bigr/\quad
  f^{\sharp} \colon X\to A^{B}\enspace.
\end{displaymath}
An example of such is via the universality of products:
\begin{displaymath}
\small\begin{array}{l}
\infer={  f^{\sharp} \colon X\longto \cat{D}(B,A)\;\text{in $\Sets$}}{  f\colon B\longto A^{X}\;\text{in a category $\cat{D}$ with arbitrary
 products}
}
\end{array}
\end{displaymath}
 where $A,B\in\cat{D}$, $X\in \Sets$ and $A^{X}$ is the $|X|$-fold
 product of $A$.


We shall use a somewhat unconventional notation of writing $X_{x}$ for an
(Eilenberg-Moore) $T$-algebra $x\colon TX\to X$. In our arguments the
monad $T$ is mostly obvious from the context, and this notional
convention turns out to be succinct and informative.




\auxproof{
\begin{definition}
  Let $R \subset X \times Y$ be a binary relation.
  We then define mappings
  $\exists_R, \forall_R \colon \pow{X} \to \pow{Y}$
  and
  $\Diamond_R, \Box_R \colon \pow{Y} \to \pow{X}$
  as follows:
  \begin{align*}
    \exists_R{X'} &= \set{y \in Y}{\exists x.\; x \mathrel{R} y \land x \in X'} \\
    \forall_R{X'} &= \set{y \in Y}{\forall x.\; x \mathrel{R} y \to   x \in X'} \\
    \Diamond_R{Y'} &= \set{x \in X}{\exists y.\; x \mathrel{R} y \land y \in Y'} \\
    \Box_R{Y'}     &= \set{x \in X}{\forall y.\; x \mathrel{R} y \to   y \in Y'} \enspace .
  \end{align*}

  The first operator $\exists_R$ takes the \emph{image} of a given set
  by $R$; the $\forall_R$ is the dual of $\exists_R$ (that is probably used less
  commonly).  The operators $\Diamond_R$ and $\Box_R$ are well-known in the
  context of modal logic. They  are called \emph{diamond modality} and
  \emph{box modality}.

\end{definition}

Notice that $\exists_R$ and $\forall_R$ work covariantly (with respect
to $R$) whereas $\Diamond_R$ and $\Box_R$ work contravariantly. In fact it holds that $\Diamond_R = \exists_{R^{-1}}$ and $\Box_R = \forall_{R^{-1}}$.
We have two pairs of adjunctions $\exists_R \dashv \Box_R$
and $\Diamond_R \dashv \forall_R$.

\begin{remark}
  If $R$ is functional, that is, there is a function $f\colon X\to Y$
  such that $x\mathrel{R}y$   if and only if $f(x)=y$,
  then we have $\Box_R = \Diamond_R = f^{-1}$
  and the situation degenerates to the familiar adjunction
  $\forall_f \dashv f^{-1} \dashv \exists_f$.
\end{remark}
}

\section{Leading Example: Nondeterministic Computation and Join- (or
 Meet-) Preservation}
\label{sec:nondet}

In this section, as a leading example, we revisit the well-known
healthiness result in Theorem~\ref{thm:healthiness-nondet-elementary}
together with its ``must''
variant. We shall
 prove the results in an abstract categorical language, paving the way to the
 general and axiomatic modeling in
 Section~\ref{sec:general-healthiness}.

\subsection{``May''-Nondeterminism}
\label{sub:diamond-modality}
In Section~\ref{sec:intro}, regarding
Theorem~\ref{thm:healthiness-nondet-elementary}, we noted
the coincidence between the healthiness condition (join-preservation)
and Eilenberg-Moore $\pow$-algebras (complete join-semilattices).
This turns out to be a deceptive coincidence---the essence lies rather
in a \emph{factorization} of the powerset monad $\pow$ by a dual
adjunction, as we shall describe.

We have a dual adjunction between $\Set$ and the category
$\CL_{\biglor}$
of complete join-semilattices and join-preserving maps.
\begin{align}
  \label{eq:set-jslat-dualadj}
  \vcenter{\xymatrix@1@C+3em{
    \Set \ar@/^/[r]^-{2^{(\place)}} \ar@{}[r]|-{\bot}
    & (\CL_{\bigvee})^{\op} \ar@/^/[l]^-{[\place, 2]_{\bigvee}}
  }}
\end{align}
It is given by a \emph{dualizing object} $2$, in the ``homming-in''
manner:
\begin{align*}
  2^{(\place)} &\colon \Set \longto (\CL_{\biglor})^{\op};\;  X \longmapsto 2^X\enspace, \\
  [\place, 2]_{\biglor} &\colon (\CL_{\biglor})^{\op} \longto \Set;\;  L \longmapsto [L, 2]_{\biglor} \enspace ;
\end{align*}
here $2$ is the poset $\{ 0 < 1 \}$,
the poset $2^X$ is the  $\abs{X}$-fold product of $2$,
and $[L, 2]_{\biglor} = \CL_{\biglor}(L, 2)$ is the set of
join-preserving maps.
This adjunction yields a monad
$X\mapsto [2^X, 2]_{\biglor}$ on $\Sets$;
the unit $\eta$ of the monad $[2^{(\place)}, 2]_{\biglor}$ is defined by $\eta_X(x) = \lambda f \ldotp f(x)$
and the multiplication $\mu_X$ is
$\mu_X(\Xi) = \lambda f \ldotp \Xi \p*{\lambda \xi \ldotp \xi(f)}$.


The following  is the first key observation.
\begin{lemma}
  \label{lem:monad-isom-jslat}
  The monad  $[2^{(\place)}, 2]_{\biglor}$ is isomorphic to the powerset
 monad $\pow$, with an isomorphism
 $\sigma \colon \pow \iso [2^{(\place)}, 2]_{\biglor}$ given by
  $\sigma_X(S) = \lambda f \ldotp \biglor_{x \in S}  f(x)$.
  \qed
\end{lemma}
\noindent
 The isomorphism in Lemma~\ref{lem:monad-isom-jslat} put us in the
 following situation.
\begin{equation}\label{eq:stateEffectForDiamondWithSigmaExplicit}
  \small
  \vcenter{\xymatrix@R-1.5em@C-1em{
  \Kl(\pow)
       \ar[r]^-{\Kl(\sigma)}_-{\cong}
  & {\Kl\bigl([2^{(\place)}, 2]_{\biglor}\bigr)}
      \ar[rr]^-{K} \ar@/^/[rd]
  && {(\CL_{\bigvee})^{\op}} \ar@/^/[ld] \\
  && {\Set} \ar@/^/[ul]^{} \ar@/^/[ur]_{} \ar@{}[ul]|{\dashv} \ar@{}[ur]|{\dashv}
  }}
\end{equation}
Here $\Kl(\sigma)$ is
the
 functor induced by the isomorphism $\sigma$ in
 Lemma~\ref{lem:monad-isom-jslat}; and $K$ is the \emph{comparison
 functor}
from the Kleisli adjunction as the
 ``initial'' factorization of a monad. See
 e.g.~\cite{MacLane71,BarrW85}.

The second key observation is that
 the top composite
 $K\co \Kl(\sigma)$---its action on arrows, precisely---coincides
 with the predicate transformer $\wpre_{\Diamond}$ in
 Theorem~\ref{thm:healthiness-nondet-elementary}.
 Indeed, identifying a binary relation $R\subseteq X\times Y$ with
 a function $X\to \pow Y$ and hence with
 a morphism $X\to Y$ in $\Kl(\pow)$, the action of $K\co \Kl(\sigma)$ can
 be concretely described as follows. The arrows on the second line are all
 in $\Sets$.
\begin{align*}\small
  \xymatrix@R=-.3em@C=2em{
    \Kl(\pow)
    \ar[r]^-{\Kl(\sigma)}
    & \Kl\bigl([2^{(\place)}, 2]_{\biglor}\bigr) \ar[r]^-{K}
    & (\CL_{\biglor})^{\op} \\
    (X \xrightarrow{R} \pow Y) \ar@{|->}[r]
    & \bigl(X \xrightarrow{\sigma_{Y}\co R} [2^{Y}, 2]_{\bigvee}\bigr) \ar@{|->}[r]
    & (2^X \xleftarrow{K(\sigma_{Y}\co R)} 2^Y)
  }
\end{align*}
Unfolding the construction of the comparison functor $K$, the function
$K(\sigma_{Y}\co R)\colon 2^{Y}\to 2^{X}$ in the end is presented as follows.
Given $f\colon Y\to 2$,
\begin{align*}
 K(\sigma_{Y}\co R)(f)
= \lambda x \ldotp (\sigma_Y \circ R)(x)(f)
= \lambda x \ldotp \biglor \set{f(y)}{x \mathbin{R} y} \enspace .
\end{align*}
This is nothing but the predicate $\wpre_{\Diamond}(R)(f)\colon X\to 2$ as
defined in Theorem~\ref{thm:healthiness-nondet-elementary}. Thus we have
established
\begin{displaymath}
 (K\co \Kl(\sigma))_{X,Y}=(\wpre_{\Diamond})_{X,Y}
 \colon
   \Kl(\pow)(X,Y) \to
   \CL_{\biglor}(2^{Y},2^{X})
\end{displaymath}
for each $X$ and $Y$.

The last key observation is that
 a comparison functor is full and faithful
in general.   The action $(K\co \Kl(\sigma))_{X,Y}$ is therefore bijective;
 hence so is $(\wpre_{\Diamond})_{X,Y}$.
This proves Theorem~\ref{thm:healthiness-nondet-elementary}.

In the arguments above the key observations have been: 1) factorization of
a monad via a \emph{dual adjunction} (Lemma~\ref{lem:monad-isom-jslat});
2) a \emph{monad map} $\sigma$ giving rise to a predicate transformer
$\wpre_{\Diamond}=K\co\Kl(\sigma)$; and 3) the role of a \emph{comparison
functor} $K$---in particular that its fullness entails healthiness.
Our general framework will be centered around these three notions (dual
adjunction, monad map and comparison), with our notion of \emph{relative
algebra} bonding them together.

\begin{remark}\label{rem:CABA}
In the above (and in Theorem~\ref{thm:healthiness-nondet-elementary}) we
 established a full and faithful functor $\Kl(\pow)\to
 (\CL_{\biglor})^{\op}$.  Cutting down its codomain, together with a
 well-known isomorphism between $\Kl(\pow)$ and the category $\Rel$ of
 sets and relations, gives us a dual equivalence $\Rel\simeq
 (\mathbf{CABA}_{\biglor})^{\op}$. Here $\mathbf{CABA}_{\biglor}$ is the
 category of complete atomic Boolean algebras and join-preserving maps
 between them. The last dual equivalence is a well-known one, found
 e.g.\ in~\cite[Section~II.9]{Halmos06} and~\cite{JonssonT51}.

 Our principal interest in this paper---motivated by healthiness in
 program logics---is in a full and faithful functor. A dual equivalence,
 in contrast, is pursued typically in the context of modal logic
 (specifically for correspondences between modal algebras and relational
 frames); see e.g.~\cite{HofmannN15}. The relevance of such equivalences
 in program logics would lie in identification of (not only programs
 but) appropriate \emph{state spaces} that realize desired predicate
 transformers. Further investigation is future work.
\end{remark}

\begin{remark}\label{rem:pitfalls}
  For a join-semilattice $L$ there is a poset isomorphism
  $L^{\op} \cong [L, 2]_{\biglor}$.
\auxproof{ Here $L^{\op}$ is the poset obtained
   from $L$ by reversing the order; and
   the isomorphism
   is  given concretely  by
   \ldotp y \not\leq x)$ (where $0\in 2$ is for false and $1\in 2$ is
   for true).  The adjunction~(\ref{eq:set-jslat-dualadj}) can also be
   stated using this isomorphism, in which case the isomorphism
   between $\pow$ and $[2^{(\place)}, 2]_{\biglor}$ becomes obvious.
}  This isomorphism $L^{\op} \cong [L, 2]_{\biglor}$ however tends to
oversimplify arguments, often leading to errors in our experience.
For a similar reason we explicitly write the isomorphism $\sigma$ in the
situation~(\ref{eq:stateEffectForDiamondWithSigmaExplicit}).
\end{remark}

\auxproof{
\begin{remark}
  \label{rem:construct-monad-map}
 \notyet
  The definition of the  map $\sigma_X$ in
  Lemma~\ref{lem:monad-isom-jslat} can alternatively be described as
  follows. This alternative definition will be the basis of our
  generalized construction in ???.

  Let us define $\sigma''_X \colon \pow X \to [2^X, 2]_{\biglor}$ as follows.
  Given $S \in \pow X$ and $f \colon X \to 2$,
  we first extend $f$ to a join-preserving map
  $f^{\sharp} \colon \pow X \to 2$.
  Here we use the join-semilattice structure of $2$
  and the fact that $\pow X$ is the free join-semilattice
  over $X$.   We define $\sigma''_X(S)(f) = f^{\sharp}(S)$;
  it is then not hard to see that $\sigma=\sigma''$.




\end{remark}
}

\subsection{``Must''-Nondeterminism}
\label{sub:box-modality}
We noted after Theorem~\ref{thm:healthiness-nondet-elementary} that a
``must''-predicate transformer $\wpre_{\Box}$ can be conceived for
nondeterministic computations, besides the ``may'' one
$\wpre_{\Diamond}$. See~(\ref{eq:boxPredTransfConcretely}).
We shall briefly describe how this variant is supported by
the same line of arguments as in
Section~\ref{sub:diamond-modality}.

The only difference from Section~\ref{sub:diamond-modality} is that we replace the dual adjunction~(\ref{eq:set-jslat-dualadj}) with
\begin{align}
  \label{eq:set-mslat-dualadj}
\small\begin{array}{l}
  \vcenter{\xymatrix@1@C+3em{
    \Set \ar@/^/[r]^-{2^{(\place)}} \ar@{}[r]|-{\bot}
    & (\CL_{\bigwedge})^{\op} \ar@/^/[l]^-{[\place, 2]_{\bigwedge}}
  }}
\end{array}
\end{align}
that is given, as before, by
$2^{(\place)} \colon \Set \to (\CL_{\bigland})^{\op};  X \mapsto
2^X$, and
$[\place, 2]_{\bigland} \colon (\CL_{\bigland})^{\op} \to \Set;\;  L \mapsto [L, 2]_{\bigland}$.
The new adjunction~(\ref{eq:set-mslat-dualadj})  factorizes the
powerset monad $\pow$, as shown much like Lemma~\ref{lem:monad-isom-jslat}.
\begin{lemma}\label{lem:monad-isom-mslat}
 A natural transformation $\sigma' \colon \pow \to [2^{(\place)}, 2]_{\bigwedge}$
 given by $\sigma'_X(S) = \lambda f \ldotp \bigland_{x \in S}  f(x)$ is
  an isomorphism of monads. \qed
\end{lemma}

Now we are in a situation that is analogous
to~(\ref{eq:stateEffectForDiamondWithSigmaExplicit}); in particular it
gives us a composite $\Kl(\pow)\xrightarrow{\sigma'}
\Kl\bigl([2^{(\place)}, 2]_{\bigwedge}\bigr)
\xrightarrow{K'}
(\CL_{\bigland})^{\op}$,
where $K'$ is a suitable comparison functor (that is full and
faithful). Working out the concrete definitions we easily observe that
\begin{align*}
\begin{array}{r}
   (K'\co \sigma')_{X,Y}\;=\;(\wpre_{\Box})_{X,Y}
  \;=\; \lambda R.\,\lambda f.\,\lambda x.\, \bigland \set{f(y)}{x \mathbin{R}
 y}
 \\
  \colon \Kl(\pow)(X,Y)\longrightarrow \CL_{\bigland}(2^{Y},2^{X})
  \enspace.
\end{array}
\end{align*}
This leads to the following analogue to
Theorem~\ref{thm:healthiness-nondet-elementary}.
\begin{theorem}[healthiness under the ``must''-nondeterminism]
 \label{thm:healthiness-nondet-elementary-box}
 Let $\varphi\colon 2^{Y}\to 2^{X}$ be a function. The following
	are equivalent.
	\begin{enumerate}
	 \item  There exists
	       $R\subseteq X\times Y$ such that
	       $\varphi=\wpre_{\Box}(R)$. Here $\wpre_{\Box}$ is
		from~(\ref{eq:boxPredTransfConcretely}).
	 \item The map $\varphi$ is meet-preserving. \qed
	\end{enumerate}
\end{theorem}

\section{State-and-Effect Triangles}
\label{sec:general-healthiness}
We continue Section~\ref{sec:nondet} and present a general and
categorical framework for establishing (possibly partial) healthiness
results.  We shall first recall the scheme of \emph{state-and-effect
  triangles}~\cite{Jacobs14CMCS,Jacobs15LMCS,Jacobs15CALCO}, and two
of its ``recipes''~\cite{Jacobs15CALCO,Hasuo15TCS} which are relevant
here.


\subsection{State-and-Effect Triangles}
\label{sub:stateAndEffectTriangles}

\paragraph{State-and-Effect Triangles, in Quantum Logic and
   Program Logic}
In the previous work~\cite{Jacobs15LMCS,Jacobs15CALCO,Jacobs14CMCS} situations called
\emph{state-and-effect triangles} have been found to be fundamental in
various examples of predicate transformers. More specifically, the
triangular scheme
dictates
how
\emph{computations}, \emph{forward state-transformer semantics} and
\emph{backward predicate-transformer semantics} are organized, in terms of
 categories, functors and a dual adjunction.
\begin{equation}\label{eq:stateEffectTriangle}\small
 \vcenter{\xymatrix@R=.8em@C=-2em{
  {\normalsize\left(\begin{array}{c}
     \text{ predicate}
     \\
		\text{transformers}
     \\
		\text{(or ``effects'')}
	 \end{array}\right)^{\op}}
     \ar@/^/[rr]
     \ar@/^/[rr];[]^{}
     \ar@{}[rr]|{\top}
 &&
  {\normalsize\left(\begin{array}{c}
    \text{state}
     \\
		\text{transformers}
     \\
		\text{(or ``states'')}
	 \end{array}\right)}
 \\
 &
  {\normalsize\left(\begin{array}{c}
		      \text{computations}
	 \end{array}\right)
}
     \ar@/_/[ru]_(.6){\quad
  {\footnotesize\begin{array}{c}
   \text{state transformer}
    \\
   \text{semantics}
   \end{array}}
  }
     \ar@/^/[lu]^(.6){\footnotesize\begin{array}{c}
   \text{predicate transformer}
    \\
   \text{semantics}
   \end{array}\quad}
}}
\end{equation}
The name ``state-and-effect triangle'' comes from the
operational study of quantum
logics; here the term
``state'' refers to a \emph{state} of a quantum system---possibly a
\emph{mixed}
state, i.e.\ a probabilistic ensemble $\sum_{i\in
       I}c_{i}\ket{\varphi_{i}}\bra{\varphi_{i}}$ over \emph{pure}
       states---and
 the term ``effect'' refers to the notion in quantum theory, i.e.\
 a convex-linear map from (quantum) states to the values in the interval
 $[0,1]$. The dual adjunction at the top
 of~(\ref{eq:stateEffectTriangle}), in such quantum settings,
 represents the duality between the
 so-called \emph{Schr\"{o}dinger} and \emph{Heisenberg} pictures of
 quantum mechanics.

In our current context of program semantics and program logics,
the term ``state'' in the state-and-effect triangles is more intuitively
understood as \emph{superposed states}, and the term ``effect'' is
understood as \emph{predicates}. See~(\ref{eq:stateEffectTriangle}). We
emphasize, in particular, that the term ``effect'' in the
state-and-effect triangles refers to the quantum notion and has nothing to do with \emph{computational
effects} in functional programs.

It is interesting that the same categorical scheme underlies quantum
logics and program logics. This is essentially because they share the
combination of \emph{logic} and \emph{dynamics}. For example, in quantum mechanics
predicates (or ``effects'') have a distinctively \emph{operational}
flavor---measurements cause projection of quantum states.

\paragraph{An Example}
Let us exhibit an example. It is based on the constructions in
Section~\ref{sub:diamond-modality}, although the triangle itself was not explicit.
\begin{equation}\label{eq:stateEffectTriangleEx1}
 \scriptsize
  \vcenter{\xymatrix@R-2.5em@C-1em{
 {\CL_{\bigvee}^{\op}}
 \ar@/^/[rr]^-{R}
 \ar@{}[rr]|{\top}
  &&
  {\EM({\pow})}
 \ar@/^/[ll]^-{L}
 \\
 &
  \Kl(\pow)
\ar@/^/[ul]^{\wpre_{\Diamond}} \ar@/_/[ur]_{K}
  }}
  \quad\text{ with $\wpre_{\Diamond}\cong LK$.}
\end{equation}
The fact that $\wpre_{\Diamond}$ is a \emph{contravariant} functor
means that the (predicate transformer) semantics expressed by it is a
\emph{backward} one.
The comparison functor $K$ from the Kleisli category
$\Kl{(\pow)}$
to the Eilenberg-Moore category $\EM(\pow)$ acts concretely as
follows:\footnote{The comparison functor $K\colon \Kl(\pow)\to\EM(\pow)$
here is different from the one $K\colon \Kl(\pow)\to
(\CL_{\bigvee})^{\op}$ in Section~\ref{sub:diamond-modality}, although
they arise from the same ``universality'' of $\Kl(\pow)$. Using the same
notation $K$ will not cause confusion.
}
\begin{multline*}
 \bigl(\,
  X\xrightarrow{f}Y\;\text{in $\Kl(\pow)$}
\,\bigr)\enspace,
 \quad\text{i.e.\ } \quad
 \bigl(\,
 X\xrightarrow{f}\pow Y\;\text{in $\Sets$}
\,\bigr)
 \\
 \stackrel{K}{\longmapsto}
 \quad
 \bigl(\,
  \pow X\xrightarrow{Kf}\pow Y, \quad
  (U\subseteq X) \longmapsto\textstyle \bigcup \{\,f(x)\mid x\in U\,\}
\,\bigr)\enspace.
\end{multline*}
The intuition is that $U\in \pow X$ is a ``superposed
state'' that indicates which states  are possibly the current
state. The triangle~(\ref{eq:stateEffectTriangleEx1}) stipulates that
$\wpre_{\Diamond}$ factors through $K$. Finally, the \emph{healthiness}
condition---that the image of $\wpre_{\Diamond}$ is characterized by
join-preservation---translates to the statement that the functor $\wpre_{\Diamond}$
in~(\ref{eq:stateEffectTriangleEx1}) is \emph{full}.

Instances of  state-and-effect triangles abound, from quantum
mechanics to computations with various notions of branching. See
e.g.~\cite{Jacobs15LMCS,Jacobs15CALCO}; later in this paper there will
be further examples, too.

\subsection{The \emph{Dual Adjunction} Recipe
}
\label{sub:dualAdjRecipe}
One ``recipe'' for state-and-effect triangles is introduced
in~\cite{Jacobs15CALCO};
we refer to it as the \emph{dual adjunction recipe}.
It works as follows.
\begin{itemize}
 \item One starts with a monad $T$ on a category $\cat{C}$, and its
       ``factorization''
       \begin{equation}\label{eq:factorizationOfAMonad}
	\scriptsize
	\xymatrix@1@C+3em{
	 {\cat{C}}
	   \ar@(lu,ld)_{T=GF}[]
	   \ar@/^/^-{F}[r]
	   \ar@{}[r]|-{\bot}
        &
	 {\cat{D}^{\op}\mathrlap{\enspace.}}
	   \ar@/^/^-{G}[l]
 	}
       \end{equation}
      We assume that the adjunction is contravariant, for the sake of
       argument.
 \item As is well-known (see e.g.~\cite{MacLane71,BarrW85}), there arise
       two comparison functors $K$ and $R$, induced by the ``universality''
       of the Kleisli and Eilenberg-Moore constructions respectively,  as below.
       \begin{equation}\label{eq:dualityRecipeStep2}
	\scriptsize
	  \vcenter{\xymatrix@R-2em@C+2em{
 {\Kl(T)}
  \ar@/^/@{-->}[r]^-{K}
  &
  {\cat{D}^{\op}}
  \ar@/^/@{-->}[r]^-{R}
&
  {\EM(T)}
 \\
 &
  {\cat{C}}
	\ar@/^/[ul]
	\ar@/^/[ul];[]
	\ar@{}[ul]|{\dashv}
	\ar@/^/[u]
	\ar@/^/[u];[]
	\ar@{}[u]|{\dashv}
	\ar@/^/[ur]
	\ar@/^/[ur];[]
	\ar@{}[ur]|{\dashv}
  }}
       \end{equation}
 \item We organize the three categories on the top in the previous diagram~(\ref{eq:dualityRecipeStep2}) as a triangle. This gives rise to
       the following situation.
       \begin{align*}
	\scriptsize
	  \vcenter{\xymatrix@R-3em@C-1em{
  {\cat{D}^{\op}}
  \ar@{->}[rr]^-{R}
  &
&
  {\EM(T)}
 \\
 &
 {\Kl(T)}
    \ar@/^/[ul]^-{K}
   \ar@/_/[ur]_-{R\co K}
  }}
       \end{align*}
 \item A left adjoint to $R$ will complete a state-and-effect
       triangle. For its existence we assume suitable equalizers in $\cat{D}$
       (hence coequalizers in $\cat{D}^{\op}$) and use a variant
       of Beck's monadicity theorem.
\end{itemize}
 A formal statement is as follows.
\begin{theorem}[the dual adjunction recipe, {\cite[Theorem~1]{Jacobs15CALCO}}]
\label{thm:recipeInJacobs15CALCO}
 Assume an adjunction~(\ref{eq:factorizationOfAMonad}) and a monad
 $T=GF$. Assume further that the category $\cat{D}$ has equalizers of
 reflexive pairs. Then we have a situation
       \begin{align*}\scriptsize
	  \vcenter{\xymatrix@R-2.3em@C-1em{
  {\cat{D}^{\op}}
  \ar@/^/@{->}[rr]^-{R}
 \ar@{}[rr]|-{\top}
  &
&
  {\EM(T)}
 \ar@/^/[ll]^-{L}
 \\
 &
 {\Kl(T)}
    \ar@/^/[ul]^-{K}
   \ar@/_/[ur]_-{R\co K}
  }}
       \end{align*}
where $LRK\cong K$. Moreover $K$ is full and faithful.
\end{theorem}
\begin{proof}
 The constructions have already been sketched in the above;
 see~\cite{Jacobs15CALCO} for details. That the comparison functor $K$
 is full and  faithful is standard; see e.g.~\cite{MacLane71,BarrW85}.
\end{proof}
Note that the dual adjunction
recipe in Theorem~\ref{thm:recipeInJacobs15CALCO}
automatically derives healthiness (that $K$ is full and
faithful). This recipe, though powerful, is also restrictive:
it obviously cannot be used to derive a non-full
predicate transformer semantics. Furthermore, the example
in~(\ref{eq:stateEffectTriangleEx1}) cannot be directly derived using
the dual adjunction recipe:
to do so we would need a slight generalization of the recipe that accommodates
a natural isomorphism $T\cong GF$---in place of the equality $T=GF$---in the
factorization~(\ref{eq:factorizationOfAMonad}). Our
generalized, ``combined'' recipe later in Section~\ref{sec:relativeAlgebraRecipe} will address these
issues.

\subsection{The  \emph{Modality} Recipe
}
\label{sub:set-pt-semantics}
Here we review the other previous recipe that we will be based on; it is
derived from the framework
of \emph{monadic predicate transformers}
from~\cite{Hasuo14,Hasuo15TCS}.\footnote{In~\cite{Hasuo14,Hasuo15TCS}
the framework is $\Pos$-based rather than $\Sets$-based. Here for
simplicity we present a $\Sets$-based framework; our generalization
later will account for the $\Pos$-based one as an instance. }
It is centered around the notion of \emph{modality}---given as an
Eilenberg-Moore
$T$-algebra $\tau\colon T\Omega\to \Omega$ over the domain $\Omega$ of
truth values---and interprets functions of the type $X\to TY$, that is,
\emph{$T$-branching computations}.



\begin{definition}[$\bbP^{\tau}$]
  \label{def:pt-semantics-in-set}
Let  $\tau \colon T {\Omega} \to \Omega$
be a $T$-algebra; it is called  a \emph{modality}. It induces
  a functor $\bbP^{\tau} \colon \Kl(T) \to \Set^{\op}$ that is defined
 by:
 $  \bbP^{\tau} X = \Omega^X$ and
 \begin{align*}
 \bbP^{\tau}\bigl(\,X \xxto{f} Y \text{\;(in $\Kl(T)$)}\,\bigr) &=
\bigl(\,\Omega^Y \xxto{\tau^{\sharp}}
  \Omega^{TY} \xxto{f^*} \Omega^X
\,\bigr)\enspace.
 \end{align*}
  Recall that $f^{*}$ denotes precomposition of $f$.
  Here $\tau^{\sharp}$ is the extension map
  that extends $h \colon Y \to \Omega$ to
  a $T$-algebra morphism $\tau^{\sharp}(h) \colon TY \to \Omega$,  via the
 bijective ``freeness'' correspondence
  \begin{align*}
     Y\longrightarrow \Omega \;\text{in $\Sets$}\quad\Bigl/\!\!\!\!\!\!\Bigr/\quad
     \Bigl(\valg{TTY}{\mu}{TY}\Bigr)
     \longrightarrow
     \Bigl(\valg{T\Omega}{\tau}{\Omega} \Bigr) \;\text{in $\EM(T)$}\enspace.
  \end{align*}

  Note that $\bbP^{\tau}(f)$ can be alternatively described as follows.
  Given $f\colon X\to TY$ (a computation) and $h\colon Y\to \Omega$ (a postcondition), the function
 $\bbP^{\tau}(f)(h)\colon X\to \Omega$ (the weakest precondition) is the composite
  \begin{math}
  X\xrightarrow{f}TY\xrightarrow{Th}T\Omega\xrightarrow{\tau}\Omega
  \end{math}.
\end{definition}

The functor $\bbP^{\tau}$ is the backward predicate transformer
semantics
induced by the modality $\tau$. It sends a state space $X$ to
the set $\Omega^X$ of predicates over $X$;
and a computation $f \colon X \to TY$ is sent to
to the (backward) predicate transformer
$\bbP^{\tau}f \colon {\Omega}^Y \to {\Omega}^X$.
The definition of $\bbP^{\tau}f$ requires a $T$-algebra
structure on $\Omega$; it determines how to interpret $T$-effects, and
hence is called a modality.

\begin{example}\label{ex:modalityAndPredTrans}
 Consider the set $2=\{0,1\}$ of the Boolean truth values; as a
 convention we identify $1$ as ``true.'' There are two
 $\pow$-algebra structures over $2$:
 \begin{align*}
  \tau_{\Box}=\bigwedge\;\colon \pow 2\to 2
  \quad\text{and}\quad
  \tau_{\Diamond}=\bigvee\;\colon \pow 2\to 2\enspace,
 \end{align*}
 where inf and sup refer to the order $0<1$. The former is
 the ``must'' modality, whereas the latter is the ``may'' one.
\end{example}

A modality $\tau\colon T\Omega\to \Omega$  gives rise to
an instance of the state-and-effect triangle.


\begin{theorem}[{the modality
 recipe,~\cite{Hasuo14,Hasuo15TCS}}]\label{thm:theModalityRecipe}
 Let $\tau\colon T\Omega\to \Omega$ be an Eilenberg-Moore algebra.
 It gives rise to the following situation, with $\bbP^{\tau}$
 factorized as  $\bbP^{\tau}\cong [\place, \Omega_{\tau}]_{T}\co K$.
  \begin{equation}\label{eq:theModalityRecipe}
   	\small
    \xymatrix@R-1.8em@C-1em{
      \Set^{\op} \ar@/^/[rr]^-{\Omega_{\tau}^{(\place)}} \ar@{}[rr]|-{\top}
      &&  \EM(T)\ar@/^/[ll]^-{[\place, \Omega_{\tau}]_{T}} \\
      & \Kl(T) \ar@/^/[ul]^-{\bbP^{\tau}} \ar@/_/[ur]_-{K}
    }
  \end{equation}
 Here $\bbP^{\tau}$ is from Definition~\ref{def:pt-semantics-in-set}, and
 $K$ is the comparison functor. The dual adjunction on the top  is
 induced by the dualizing object $\Omega_{\tau}$.\footnote{Recall our
 notational convention that an Eilenberg-Moore algebra $\tau\colon
 T\Omega\to \Omega$ is denoted by $\Omega_{\tau}$. See
 Section~\ref{sec:intro}. }
Specifically, the
 functor $[\place, \Omega_{\tau}]_{T}=\EM(T)(\place, \Omega_{\tau}) $
 is the homset functor; and $\Omega_{\tau}^X$  for a set $X$ is the
 $X$-fold product of the $T$-algebra $\Omega_{\tau}$. The latter is
 explicitly given by the transpose of
 \begin{align*}
  X \stackrel{\id^{\sharp}}{\to} \Set(\Omega^X, \Omega)
  \xrightarrow{\!T_{\Omega^{X},\Omega}\!} \Set\bigl(T(\Omega^X), T{\Omega}\bigr)
  \stackrel{\tau_*}{\to} \Set\bigl(T(\Omega^X), \Omega\bigr)
 \end{align*}
 where $\id^{\sharp}$ is the transpose of
 the identity $\id \colon \Omega^X \to \Omega^X$
 and $T_{\Omega^{X},\Omega}$ is the action of $T$ on homsets.
\qed
\end{theorem}

\section{The ``Relative Algebra'' Recipe for State-and-Effect Triangles}
\label{sec:relativeAlgebraRecipe}
We unify the two recipes (\emph{dual adjunction} and \emph{modality}) to render a general one. It is called the
\emph{relative algebra recipe}, because of the  role played by
our notion of \emph{relative Eilenberg-Moore algebra}.

\subsection{Relative Eilenberg-Moore Algebra}
We shall introduce the notion of relative Eilenberg-Moore algebra for a
monad $T$ on $\Sets$ and a category $\cat{D}$ with small
products. Notably its carrier object is an object of $\cat{D}$; hence
what we do is arguably to interpret a monad $T$ on $\Sets$ over a
different category $\cat{D}$.

\begin{remark}
 We expect further generalization is possible. The developments below
 bear a strong enriched flavor; and we envisage a general framework
 where a $\V$-monad $T$ on an SMCC $\V$ is interpreted over an arbitrary
 $\V$-enriched category $\cat{D}$. Working out the precise statements is
 future work.
\end{remark}

Let $\cat{D}$ be a category with arbitrary products.
For each object $A\in\cat{D}$ there is a  dual adjunction,
with $A$ playing the role of a dualizing object.
\begin{equation}\label{eq:dual-adj-cotensor-hom}
 \vcenter{  \xymatrix@R-1.2em@C-1em{
    \Set \ar@/^/[rr]^-{A^{(\place)}} \ar@{}[rr]|-{\bot}
    &&  \cat{D}^{\op} \ar@/^/[ll]^-{\cat{D}(\place, A)}
  }
},
\quad\text{by}\quad
\vcenter{\infer={X\longto \cat{D}(B, A) \text{ in $\Sets$}}{B \longto
A^{X}\text{ in $\cat{D}$}}}\enspace.
\end{equation}
This is much like
in~(\ref{eq:set-jslat-dualadj}); recall that $A^{X}$ denotes the
$|X|$-fold
product of $A\in\cat{D}$ (i.e.\ a \emph{cotensor}, in the enriched
terms).

This adjunction induces
a continuation-like monad $\cat{D}(A^{(\place)}, A)$.

\begin{definition}[$\cat{D}$-relative $T$-algebra]
  \label{def:relative-emcat}
  Let $T$ be a monad on $\Sets$, and $\cat{D}$ be a category with small
 products.
A \emph{$\cat{D}$-relative $T$-algebra} is
\begin{displaymath}
 \text{a pair} \quad
 \bigl(\,A,\; \alpha\colon T \to \cat{D}(A^{(\place)}, A)\,\bigr)
\end{displaymath}
 of an object $A\in\cat{D}$ and a monad map $\alpha$ from $T$ to the
 continuation-like monad
 $\cat{D}(A^{(\place)}, A)$ from~(\ref{eq:dual-adj-cotensor-hom}).

 A \emph{morphism} of $\cat{D}$-relative $T$-algebras, say from
 $(A,\alpha)$ to $(B,\beta)$, is a morphism $f \colon A \to B$ in
 $\cat{D}$ such that the following diagram commutes for each $X \in
 \Set$.
      \begin{align*}
        \xymatrix@R-1.8em@C+2em{
          A^X \ar[r]^-{\alpha_X^{\sharp}} \ar@<-0.3em>[d]_-{f_*}
          & A^{TX} \ar@<-0.5em>[d]^-{f_*} \\
          B^X \ar[r]^-{\beta_X^{\sharp}}
          & B^{TX}
        }
      \end{align*}
  Here $\alpha^{\sharp}_{X}$ is induced canonically from
 $\alpha_{X}\colon TX\to \cat{D}(A^{X}, A)$, via the bijective
 correspondence in~(\ref{eq:dual-adj-cotensor-hom}) (the
 universality of products).

 $\cat{D}$-relative $T$-algebras, together with
 their morphisms, form a category $\EM(T; \cat{D})$ that we call the
 \emph{$\cat{D}$-relative Eilenberg-Moore category of $T$}. It comes with
 an obvious forgetful functor to $\cat{D}$:
 \begin{equation}\label{eq:forgetfulFunctorForRelativeAlg}
   U_{\cat{D}}\; \colon\; \EM(T;\cat{D}) \longto \cat{D} \enspace.
 \end{equation}
\end{definition}

There are many questions to be asked about relative algebras, for
example if $U_{\cat{D}}$ has a left adjoint. These questions are left as
future work: they seem to be best studied in conjunction with Lawvere
theories, and doing so deviates from the current paper's focus.

 We shall still
show that relative algebras  generalize the usual
notion of Eilenberg-Moore algebra.
We rely on the following folklore result on: algebras, and monad maps to
continuation-like monads. It is used e.g.\ in~\cite{Kelly80,KellyP93}.
\begin{proposition}
  \label{prop:emalg-monadmap}
  Let $\cat{C}$ be a complete category and $T$ be a monad on
 $\cat{C}$. For each object $A\in \cat{C}$,
  there is a canonical bijective correspondence between:
 1)   $T$-algebras $\hat{\alpha} \colon T A \to A$ with $A$ being their carrier objects;
  and 2) monad maps $\alpha \colon T \to A^{\cat{C}(\place, A)}$.

  The concrete correspondence is given by:
  \begin{math}
  \alpha_{X}=\bigl\langle \, TX\xrightarrow{Tf}TA\xrightarrow{\hat{\alpha}}A\,\bigr\rangle_{f\in
   \cat{C}(X,A)}
  \end{math} and
  \begin{math}
    \hat{\alpha} =
    \bigl(\,
   TA\xrightarrow{\alpha_{A}}A^{\cat{C}(A,A)}\xrightarrow{\pi_{\id_{A}}}
    A
\,\bigr)
  \end{math}.
  Moreover, $f \colon A \to B$ is a $T$-algebra morphism from $(A,\hat{\alpha})$ to $(B,\hat{\beta})$
  if and only if the following diagram commutes.
  \begin{align*}
    \xymatrix@R-2em@C+1em{
      \cat{C}(X, A) \ar[r]^-{\alpha^{\sharp}_X} \ar[d]^-{f_*}
      & \cat{C}(TX, A) \ar[d]^-{f_*} \\
      \cat{C}(X, B) \ar[r]^-{\beta^{\sharp}_X}
      & \cat{C}(TX, B)
    }
  \end{align*}
 Here $a^{\sharp}_{X}$ is defined analogously to
 Definition~\ref{def:relative-emcat}.
 \qed
\end{proposition}
This result and Definition~\ref{def:relative-emcat} yields the
following.
There we also need the isomorphism
 $\Sets(A^{X},A)\cong A^{\Sets(X,A)}$ that identifies homsets and
 cotensors. This is available since $\Sets$ is self-enriched.
\begin{corollary}\label{cor:SetsRelativeIsOrdinary}
 Let $T$ be a monad on $\Sets$. We have an isomorphism
 $\EM(T;\Sets)\cong\EM(T)$.
\qed
\end{corollary}

\begin{remark}
 There is a Lawvere theory-like intuition behind
 Proposition~\ref{prop:emalg-monadmap} (from which we came up with
 Definition~\ref{def:relative-emcat}). Given an algebra $\hat{\alpha} \colon
 TA\to A$,
 the corresponding monad map $\alpha_{X}\colon TX\to A^{\cat{C}(X,A)}$
 is understood as: ``given an algebraic/syntactic term $t\in TX$ with variables
 from $X$, and a valuation $V\colon X\to A$, the element $\alpha_{X}(t)(V)\in
 A$ is how $t$ is interpreted under $V$ (interpreting variables) and
 $\hat{\alpha}$
 (interpreting algebraic operations).''
 \end{remark}

 Lawvere theories are interpretation-free---hence
 ``syntactic''---presentations of algebraic structures. They are
 therefore subject to interpretation in \emph{any} category $\cat{D}$
 with finite products; see e.g.~\cite{hyland2007category}. In
 contrast, monads---although their equivalence to Lawvere theories is
 well-known, see e.g.~\cite{LackP09}---are always tied to their base
 category. Our notion of $\cat{D}$-relative $T$-algebra is how to
 ``interpret'' the algebraic structure embodied as a monad $T$ (on
 $\Sets$) on another category $\cat{D}$.\footnote{We speculate that,
   when a monad $T$ is bounded, our notion of relative $T$-algebra
   coincides with the models of the Lawvere theory $\cat{L}_{T}$
   induced by $T$. We note however that relative $T$-algebras can be
   defined even for unbounded $T$. We need this feature, too, since we
   deal with unbounded monad like the powerset monad $\pow$. }

\begin{example}
  \label{ex:example-of-relative-emcat}
  Let $\List$ denote the list monad on $\Sets$,
 whose Eilenberg-Moore algebras are monoids.
  For $\cat{D} = \Top$, the category of topological spaces and continuous
 maps,  the  category $\EM(\List; \Top)$
  is exactly the category of \emph{topological monoids}.
  Similarly for $\cat{D} = \Pos$, the category of posets and monotone maps, the  category
  $\EM(\List; \Pos)$ is that of \emph{ordered monoids}.
  The same phenomena can be observed for many other monads $T$ and
 categories $\cat{D}$.
\end{example}

We exhibit a  \emph{change-of-base} result.
In the case of Lawvere theories,
we can map a $\cat{D}$-model of a theory to a $\cat{D'}$-model
along a (finite) product-preserving functor
$H \colon \cat{D} \to \cat{D'}$.

\begin{proposition}
  \label{prop:change-of-base}
  Let $\cat{D}$, $\cat{D'}$ be categories with small products and
  $H \colon \cat{D} \to \cat{D'}$ be a product-preserving functor.
  Then $H$ canonically lifts to a functor
  $\bar{H} \colon \EM(T; \cat{D}) \to \EM(T; \cat{D'})$,
 with
  \begin{equation}\label{eq:change-of-base-lifting}
\vcenter{    \xymatrix@R-1.8em{
      \EM(T;\cat{D}) \ar[r]^-{\bar{H}} \ar[d]_-{U_{\cat{D}}}
      & \EM(T;\cat{D'}) \ar[d]^-{U_{\cat{D'}}} \\
      \cat{D} \ar[r]^-{H}
      & \cat{D'} \mathrlap{\enspace.}
    }
}  \end{equation}
  The functor $\bar{H}$ preserves arbitrary products. Moreover,
  if $H$ is faithful, so is $\bar{H}$.
 \qed
\end{proposition}


\auxproof{
The following  alternative characterization of
 $\cat{D}$-relative $T$-algebras will be used later.
Its proof is straightforward.
\begin{proposition}
  A natural transformation
  $\tau \colon T \to \cat{D}(A^{(\place)}, A)$
  is a monad map if and only if
  $\alpha^{\sharp} \colon A^{X} \to A^{TX}$
  makes the following diagrams commute: \memo{really?}
  \begin{align*}
    \xymatrix@R-1.8em{
      A^X \ar@{=}[rd] \ar[r]^-{\alpha_X^{\sharp}}
      & A^{TX} \ar@<-0.4em>[d]^-{\eta_X^*} \\
      & A^X
    }
   \quad
    \xymatrix@R-1.8em{
      A^X \ar[r]^-{\alpha^{\sharp}_X}
        \ar@<-0.4em>[d]_-{\alpha^{\sharp}_X}
      & A^{TX} \ar@<-0.8em>[d]^-{\mu^*_X} \\
      A^{TX} \ar[r]^-{\alpha^{\sharp}_{TX}}
      & A^{T^2 X} \\
    }
  \end{align*}
\end{proposition}

\vspace{-1.6em}
\qed
}

\subsection{A State-and-Effect Triangle via Relative Algebras}
\label{sub:state-effect-triangle-relative}

In the ``modality'' recipe in Section~\ref{sub:set-pt-semantics},
the key to  a dual adjunction
between $\Set$ and $\EM(T)$
was to use a $T$-algebra as a dualizing object.
We shall now  extend this from $\Sets$ to a general category $\cat{D}$,
using
 a $\cat{D}$-relative $T$-algebra
in place of a $T$-algebra.

\begin{theorem}[the relative algebra recipe]
  \label{thm:state-effect-triangle-relative}
  Let $\OmegaD\in\cat{D}$ be an object in a complete category $\cat{D}$,
and
\begin{displaymath}
  \bar{\Omega} \;=\; \bigl(\,\OmegaD, \,\tau\colon T\to
  \cat{D}(\OmegaD^{(\place)}, \OmegaD)\,\bigr)
\end{displaymath}
be a $\cat{D}$-relative
 $T$-algebra.
 This yields a state-and-effect triangle:
  \begin{equation}\label{eq:state-effect-triagle-relative}
\vcenter{
 \xymatrix@R-1.2em@C-1em{
      \cat{D}^{\op} \ar@/^/[rr]^-{[\place, \bar{\Omega}]_\cat{D}}
        \ar@{}[rr]|-{\top}
      &&  \EM(T)\ar@/^/[ll]^-{[\place, \bar{\Omega}]_{T}} \\
      & \Kl(T) \ar@/^/[ul]^-{\bbP^{\tau}} \ar@/_/[ur]_-{K}
      & 
    }
}
\quad\text{with $\bbP^{\tau} \cong [-, \bar{\Omega}]_{T} \co K$.}
  \end{equation}
  Here $K$ is the comparison functor.
  The other three functors
are defined as follows.
\begin{itemize}
 \item ($[\place, \bar{\Omega}]_{\cat{D}}$)
 For each $D \in \cat{D}$,
  the object $[D, \bar{\Omega}]_{\cat{D}}$
  is the  set $\cat{D}(D, \OmegaD)$ equipped with
  a $T$-algebra structure
   $\zeta_D$  defined by
  \begin{equation}
    {T\bigl(\cat{D}(D, \OmegaD)\bigr)} \stackrel{\tau}{\to}
    \cat{D}(\OmegaD^{\cat{D}(D, \OmegaD)}, \OmegaD)
    \stackrel{(\id^{\sharp})^*}{\to}
    \cat{D}(D, \OmegaD) \enspace .
  \end{equation}
  The last arrow precomposes $\id^{\sharp}\colon D\to
       \OmegaD^{\cat{D}(D, \OmegaD)}$.

  For a $\cat{D}$-morphism $k \colon D \to E$, the $T$-algebra morphism
  $k^* = [k, \bar{\Omega}]_{\cat{D}}
  \colon [E, \bar{\Omega}]_{\cat{D}} \to [D, \bar{\Omega}]_{\cat{D}}$
  is defined by the precomposition map
  $k^* \colon \cat{D}(E, \OmegaD) \to \cat{D}(D, \OmegaD)$
  between the carrier sets.
 \item  ($[\place, \bar{\Omega}]_{T}$)
Given a $T$-algebra $A_a = (A, a \colon TA \to A)$,
  the $\cat{D}$-object $[A_a, \bar{\Omega}]_{T}$
  is defined as the equalizer of
  $a^*, \tau^{\sharp}_A \colon \OmegaD^A \tto \OmegaD^{TA}$; see the top
	row of~(\ref{eq:201601171013}) below.
 Given a morphism $f \colon A_a \to B_b$ of $T$-algebras,
  a $\cat{D}$-morphism
  $f^* = [f, \bar{\Omega}]_{T} \colon [B_b, \bar{\Omega}]_{T} \to [A_a, \bar{\Omega}]_{T}$
  is induced by the universality of an equalizer, as below.
  \begin{equation}\label{eq:201601171013}
    \vcenter{
      \xymatrix@R-1.5em{
        [A_a, \bar{\Omega}]_{T} \ar[r]^{\mathrm{eq}}
        & \OmegaD^A \ar@<.5ex>[r]^{a^*} \ar@<-.5ex>[r]_{\tau_A^{\sharp}}
        & \OmegaD^{TA} \\
        [B_b, \bar{\Omega}]_{T} \ar[r]^{\mathrm{eq}} \ar@{-->}[u]^{f^*}
        & \OmegaD^B \ar@<.5ex>[r]^{b^*} \ar@<-.5ex>[r]_{\tau_B^{\sharp}}
                    \ar[u]^{f^*}
        & \OmegaD^{TB} \mathrlap{\enspace .} \ar[u]_{(Tf)^*}
      }
    }
  \end{equation}
 \item ($\bbP^{\tau}$)
 $\bbP^{\tau} \colon \Kl(T) \to \cat{D}^{\op}$ is
  given by: $\bbP^{\tau}(X) = \OmegaD^X$, and
  \begin{align*}
    \bbP^{\tau}\bigl(\,X \xxto{f} TY\text{\;(in $\Kl(T)$)}\,\bigr)
      = \bigl(\,\OmegaD^{Y} \stackrel{\tau^{\sharp}}{\to}
          \OmegaD^{TY}  \stackrel{f^*}{\to} \OmegaD^{X}\,\bigr)  .
   \tag*{\qed}
  \end{align*}
\end{itemize}
\end{theorem}

\begin{remark}
  The notations $[A_a, \bar{\Omega}]_{T}$ and $[D,
  \bar{\Omega}]_{\cat{D}}$ are sort of abusive, because $\bar{\Omega}$
  is not a $T$-algebra or a $\cat{D}$-object. The notations reflect
  the dual nature of a $\cat{D}$-relative $T$-algebra, in the sense
 that is precisely described in the above. Later in
 Section~\ref{subsec:relative-alg-concrete} we develop this point (and
 notations) more systematically.
\end{remark}

The third \emph{relative algebra} recipe in
Theorem~\ref{thm:state-effect-triangle-relative} combines the previous
two recipes. Indeed, the \emph{modality} recipe in
Section~\ref{sub:set-pt-semantics} is a special case, much like usual
$T$-algebras are special cases of relative $T$-algebras
(Corollary~\ref{cor:SetsRelativeIsOrdinary}).  The current
generalization allows us to have a  category $\cat{D}$---possibly other than
$\Sets$---at the top-left of a state-and-effect
triangle. See~(\ref{eq:theModalityRecipe})
and~(\ref{eq:state-effect-triagle-relative}).

Regarding the relationship to the \emph{dual adjunction} recipe in
Section~\ref{sub:dualAdjRecipe}---that automatically ensures
healthiness, see Theorem~\ref{thm:recipeInJacobs15CALCO}---we have the following cornerstone result
towards analysis of general healthiness conditions. Note that the assumption
$T=GF$ in~(\ref{eq:factorizationOfAMonad}) translates, in the context of Theorem~\ref{thm:state-effect-triangle-relative}, to the condition
that $\tau\colon T\to
\cat{D}(\OmegaD^{(\place)}, \OmegaD)$ is the identity.




\begin{theorem}[categorical (partial) healthiness condition]\label{thm:partialHealthiness}
Let $X,Y\in\Sets$. In the setting of
 Theorem~\ref{thm:state-effect-triangle-relative}:
 \begin{itemize}
  \item if $\tau_{Y}\colon TY\to
  \cat{D}(\OmegaD^{Y}, \OmegaD)$ is injective,
	 the functor $\bbP^{\tau}$'s action
\begin{math}
   \bbP^{\tau}_{XY} \colon \Kl(T)(X, Y) \to \cat{D}(\OmegaD^Y, \OmegaD^X)
\end{math}
   is injective;
  \item if $\tau_{Y}$ is surjective, so is $   \bbP^{\tau}_{XY} $.
\end{itemize}
 It follows that, if $\tau\colon T\to \cat{D}(\OmegaD^{(\place)},
 \OmegaD)$ is a natural isomorphism, the functor $\bbP^{\tau}$ is full
 and faithful.
\qed
\end{theorem}
 \noindent
 The last corollary accounts for (part of)
 Theorem~\ref{thm:healthiness-nondet-elementary}, generalizing the
 arguments in Section~\ref{sub:diamond-modality}. The proof of
 Theorem~\ref{thm:partialHealthiness} is like the proof
 of~\cite[Theorem~IV.3.1]{MacLane71}, giving the correspondence
 between monic/epic (co)units and fullness/faithfulness of adjoints.


\subsection{Relative Algebras over a Concrete Category}
\label{subsec:relative-alg-concrete}
We have obtained the third, unified recipe for state-and-effect
triangles in Theorem~\ref{thm:state-effect-triangle-relative}, together
with a general healthiness result
(Theorem~\ref{thm:partialHealthiness}). The remaining piece towards the
full coverage of healthiness results like
Theorem~\ref{thm:healthiness-nondet-elementary} is: how specific predicate
transformer semantics---specified by a concrete modality (like
$\wpre_{\Diamond}$ via $\Diamond$)---is related to constructs in the
general recipe in Theorem~\ref{thm:state-effect-triangle-relative}.

To fill this missing piece we shall study
a situation where $\cat{D}$ is \emph{concrete},
by which we specifically mean that: 1) we have a faithful ``forgetful''
functor $V \colon \cat{D} \to \Set$; and 2) the functor $V$ preserves
small limits. Examples are: the Eilenberg-Moore category $\EM(T)$ of a monad $T$
on $\Sets$; the categories $\Top$ and $\Pos$; and other categories of
 ``sets with additional
structures.''

By  Proposition~\ref{prop:change-of-base} and
Corollary~\ref{cor:SetsRelativeIsOrdinary},
the functor $V \colon \cat{D} \to \Set$ lifts
to $\bar{V} \colon \EM(T; \cat{D}) \to \EM(T)$.
This gives rise to the following.
\begin{equation}\label{eq:collaborationOfDAndEMT}
\vcenter{  \xymatrix@R-3.5em@C-1em{
    & \EM(T; \cat{D}) \ar[ld]_-{U_{\cat{D}}}
          \ar[rd]^-{\bar{V}} \\
    \cat{D} \ar[rd]_-{V}
    && \EM(T) \ar[ld]^-{U} \\
    & \Set
  }
}
\quad\text{Here $U_{\cat{D}}$ is from~(\ref{eq:forgetfulFunctorForRelativeAlg}).}
\end{equation}
The diagram~(\ref{eq:collaborationOfDAndEMT}) is skewed (compared
    to~(\ref{eq:change-of-base-lifting})) to convey the intuition that:
an object in $\EM(T; \cat{D})$ is a set equipped both with a
    $\cat{D}$-structure
and with a $T$-algebra structure, in a compatible manner.
The developments below are aimed at formalizing this intuition.

\begin{notation}[$A, A_{\cat{D}}, A_{\hat{\alpha}},\bar{A}$]
\label{notation:theMessyNotation}
From now on we adopt a  notational convention
of writing: $A$ for a set; $A_{\cat{D}}$ for an object in $\cat{D}$ such
that $V(A_{\cat{D}})=A$; $A_{\hat{\alpha}}$ for a $T$-algebra
 $\hat{\alpha}\colon TA\to A$ (hence $U(A_{\hat{\alpha}})=A$); and
 $\bar{A}\in \EM(T; \cat{D})$ for a relative $T$-algebra such that
 $V U_{\cat{D}}\bar{A}=U\bar{V}\bar{A}=A$. See below, and compare it to~(\ref{eq:collaborationOfDAndEMT}).
 \begin{equation}\label{eq:collaborationOfDAndEMTObj}
 \vcenter{  \xymatrix@R-3em@C+2em{
     & \bar{A} \ar@{|->}[dd]^(.3){U\bar{V}=VU_{\cat{D}}} \\
     A_{\cat{D}} \ar@{|->}[rd]_-{V}
     && A_{\hat{\alpha}} \ar@{|->}[ld]^-{U} \\
     & A
   }
 }\end{equation}
 This convention, though admittedly confusing at first sight, follows some literature
 on dualities (such as~\cite{ClarkD98}) and allows us to describe our
 technical developments in a succinct  manner. We emphasize that
\begin{center}
  \underline{\emph{fixing $A$ does not fix $A_{\cat{D}}$, $A_{\hat{a}}$ or $\bar{A}$.}}
\end{center}
\end{notation}



The following characterization of $\cat{D}$-relative $T$-algebras---it
assumes that  $\cat{D}$ is concrete---embodies the intuition that they
are ``$T$-algebras whose algebraic structures are compatible with
$\cat{D}$-structures.'' This is much like a topological monoid is a
monoid whose multiplication is continuous; see
Example~\ref{ex:example-of-relative-emcat}.
\begin{proposition}
  \label{prop:lifting-condition}
  Let $A_{\hat{\alpha}} = \p[\big]{A, \hat{\alpha} \colon TA \to A}$
  be a $T$-algebra and
  $A_{\cat{D}}\in\cat{D}$ be such that $V A_{\cat{D}} = A$.
  The following are equivalent.
  \begin{enumerate}
   \item  There exists
  a $\cat{D}$-relative $T$-algebra $\bar{A}$
  such that $\bar{V} (\bar{A}) = A_{\hat{\alpha}}$
  and $U_{\cat{D}} \bar{A} = A_\cat{D}$. See below.
   \begin{equation}\label{eq:compatibility-D-EMT}
  \vcenter{  \xymatrix@R-3.3em{
      & \bar{A} \ar@{|-->}[ld]_-{U_{\cat{D}}} \ar@{|-->}[rd]^-{\bar{V}} \\
      A_{\cat{D}} \ar@{|->}[rd]_-{V}
      && A_{\hat{\alpha}} \ar@{|->}[ld]^-{U} \\
      & A
    }
  }\end{equation}
   \item  The following
  \emph{lifting condition} holds:
  the monad map $\alpha \colon T \to \Set(A^{(\place)}, A)$
  induced by $\hat{\alpha}$ (Proposition~\ref{prop:emalg-monadmap})
	  factors through
  $V \colon \cat{D}(A_{\cat{D}}^{(\place)}, A_{\cat{D}}) \to
	  \Set(A^{(\place)}, A)$, as in
  \begin{equation}\label{eq:lifting-condition}
\vcenter{    \xymatrix@R-2em{
      & \cat{D}(A_{\cat{D}}^{X}, A_{\cat{D}}) \ar[d]^-{V} \\
      TX \ar[r]_-{\alpha_X} \ar@{-->}[ur]^-{\bar{\alpha}_X}
      & \Set(A^{X}, A)\mathrlap{\enspace.}
    }
}  \end{equation}
   The latter is more concretely stated as follows:
    for each  $X \in \Set$ and $t \in TX$,
     the function $(\alpha_X)(t) \colon A^X \to A$
  lifts to a $\cat{D}$-morphism
  $(\bar{\alpha}_X)(t) \colon A_{\cat{D}}^X \to A_{\cat{D}}$.
  \end{enumerate}
  If the conditions hold, we say that
  $A_{\cat{D}}$ and $A_{\hat{\alpha}}$ are \emph{compatible}.
  \qed
\end{proposition}
This result means: to render $A\in\Sets$ into
 a $\cat{D}$-relative $T$-algebra, it suffices to find a $T$-algebra
 structure and a
$\cat{D}$-structure and then to check the above lifting condition.
The lifting condition in Proposition~\ref{prop:lifting-condition}
is a direct generalization of the \emph{monotonicity condition}
(precisely its pointwise version) used in \cite{Hasuo14, Hasuo15TCS};
when $\cat{D} = \Pos$ we get the original monotonicity condition.

Under a further assumption that $T$ is finitary, we can restrict the
required check to finite sets.


\begin{proposition}
Assume the setting of Proposition~\ref{prop:lifting-condition} and $T$ is finitary.
 $A_{\hat{\alpha}}$ and $A_{\cat{D}}$ are compatible
  if and only if the lifting condition~(\ref{eq:lifting-condition}) holds
 for any natural number $n$ in place of $X$.
 \qed
\end{proposition}

\noindent
It follows that, in case the monad $T$ is  induced by some known
algebraic specification $(\Sigma,E)$,   checking
 the lifting condition can further be restricted to ``basic
operations''  $\sigma \in \Sigma$.
For instance, a monoid $(\Omega, \star, e)$ and
 $\OmegaD \in \cat{D}$
satisfy the lifting condition
if and only if both
the multiplication $\star \colon \Omega \times \Omega \to \Omega$ and
the unit $e \colon 1 \to \Omega$ lift to $\cat{D}$-morphisms.


We can exploit the construction in Theorem~\ref{thm:state-effect-triangle-relative}
when $\cat{D}$ is concrete.
Assume we have a  functor $V \colon \cat{D} \to \Set$ that is faithful
and limit-preserving.
Given a $\cat{D}$-object $D$, the $T$-algebra $[D, \bar{\Omega}]_{\cat{D}}$
can be understood as a homset $\cat{D}(D, \bar{\Omega})$
with the \emph{pointwise} $T$-algebra structure along $\Omega_{\tau}$
i.e.\ we regard it as a subalgebra of the product $\Omega_{\tau}^{VD}$.
Similarly we can see $[A_a, \bar{\Omega}]_{T}$ as the $\cat{D}$-object
whose ``carrier set'' is $\EM(T)(A_a, \Omega_{\tau})$ with the pointwise $\cat{D}$-structure.
It is precisely stated as follows.

\begin{proposition}
  \label{prop:state-pred-duality-forgetful}
  In the situation of Theorem~\ref{thm:state-effect-triangle-relative}
  we have $U \after [D, \bar{\Omega}]_{\cat{D}} = \cat{D}(D,
 \bar{\Omega})$.

  Furthermore, assume we have a limit-preserving functor $V \colon \cat{D} \to \Set$
  and let $\bar{V}(\bar{\Omega}) = \Omega_{\tau}$.
  Then $[D, \bar{\Omega}]_{\cat{D}}$ is a $T$-subalgebra of
 $\Omega_{\tau}^{VD}$---meaning that the algebraic structure of the
 former is a pointwise extension of $\tau$---and we have $V \co [\place,
 \bar{\Omega}]_{T} \cong \EM(T)(\place, \Omega_{\tau})$.
\qed
\end{proposition}
\noindent
In the latter setting of
Proposition~\ref{prop:state-pred-duality-forgetful} where
 $\cat{D}$ is ``concrete'' with $V\colon \cat{D}\to\Sets$, let
 $\bar{\Omega} = (\OmegaD, \bar{\tau}) \in \EM(T;\cat{D})$ and
 $\bar{V}\bar{\Omega} = \Omega_{\tau} = (\Omega, \tau \colon T{\Omega} \to \Omega)$
be its underlying $T$-algebra. These data give rise to
two different
 predicate transformer semantics:
 one is $\bbP^{\bar{\tau}} \colon \Kl(T) \to \cat{D}^{\op}$
from the relative algebra $\bar{\Omega}$ via the relative algebra recipe (Theorem~\ref{thm:state-effect-triangle-relative}); and
the other is $\bbP^{\tau} \colon \Kl(T) \to \cat{\Set}^{\op}$
from the (ordinary) $T$-algebra $\Omega_{\tau}$ via the modality recipe
(Theorem~\ref{thm:theModalityRecipe}). Between these we have the
following correspondence, as we announced in the beginning of
Section~\ref{subsec:relative-alg-concrete}.
\begin{proposition} \label{prop:01190015}
  In the situation of Theorem~\ref{thm:state-effect-triangle-relative},
 additionally assume that  $V\colon \cat{D}\to\Sets$ is a faithful and
 limit-preserving functor. Then, in terms of the above notations we have
  $\bbP^{\tau} \cong V \after \bbP^{\bar{\tau}}$.
  \begin{equation}\label{eq:state-effect-triagle-concrete}
  \vcenter{
   \xymatrix@R-2.5em@C-1em{
        \cat{D}^{\op} \ar[rr]^-{V}
        &&  \Set^{\op} \\
        & \Kl(T) \ar@/^/[ul]^-{\bbP^{\bar{\tau}}} \ar@/_/[ur]_-{\bbP^{\tau}}
      }
  }
  \end{equation}
\end{proposition}

 We have required $V\colon \cat{D}\to \Sets$ to be limit-preserving; in fact it mostly suffices
 to assume \emph{product-preservation}. In that case the only thing that
 fails is the isomorphism $V \co [\place,
 \bar{\Omega}]_{T} \cong \EM(T)(\place, \Omega_{\tau})$ in
 Proposition~\ref{prop:state-pred-duality-forgetful}, which is the result
 that connects the top-right corner of the two (\emph{relative algebra} and
 \emph{modality}) triangles. Proposition~\ref{prop:01190015} is
 only concerned about the top-left and bottom corners, and survives
 under $V$ that preserves only products.

%


\subsection{Finitary Predicate Transformers}
\label{subs:example-finitary}
 The categorical results so far for healthiness
(Theorem~\ref{thm:state-effect-triangle-relative}
and~\ref{thm:partialHealthiness}) are not enough for some instances of
healthiness results, as we will see
 in the examples of Section~\ref{sec:prob-examples}. Specifically, besides
 the \emph{structural} aspects covered by those results, we need
 to take account of  \emph{sizes}.

Throughout  Section~\ref{subs:example-finitary} we adopt the setting in
Proposition~\ref{prop:01190015}, i.e.\ the relative algebra recipe
  with a faithful and limit-preserving $V\colon \cat{D}\to \Sets$.
In particular we have
a $\cat{D}$-relative $T$-algebra $\bar{\Omega} = (\OmegaD, \bar{\tau})$,
and  $\Omega_{\hat{\tau}} = \bar{V} \bar{\Omega}$ as its underlying
$T$-algebra. Recall the correspondence between $\hat{\tau}\colon
T\Omega\to \Omega$ and a monad map $\tau$ (Proposition~\ref{prop:emalg-monadmap}).


The key observation is the following lemma.
\begin{lemma}
  \label{lem:finitary-factorization}
  Let $X$ be a set and $t \in TX$.
  If $T$ is finitary, the map $(\tau_X)(t) \colon \Omega^{X} \to \Omega$
  factors through a precomposition map $s^* \colon \Omega^{X} \to \Omega^{X'}$
  for some \emph{finite} subset $s \colon X' \into X$ of $X$
  i.e.\ there exists $\phi' \colon \Omega^{X'} \to \Omega$
  such that $(\tau_X)(t) = \phi' \co s^*$.
 \qed
\end{lemma}

We formulate a size restriction on predicate transformers.
\begin{definition}[finitary predicate transformer]
 \label{def:finitaryPredTrasn}
  A predicate transformer $\phi \colon \Omega^Y \to \Omega^X$ is
 \emph{finitary}
  if for each $x \in X$ there exists a finite subset $s \colon Y' \into
 Y$
  such that $\pi_{x}\co \varphi$ factors through the precomposition
 $s^{*}$. See below.
  \begin{equation}
    \vcenter{
    \xymatrix@R-2em{
      \Omega^{Y'} \ar@{->}[drr]^{\exists\phi'} \\
      \Omega^{Y} \ar[u]^{s^*} \ar[r]_{\phi}
      & \Omega^{X} \ar[r]_{\pi_x}
      & \Omega
    }
    }
  \end{equation}
\end{definition}

\begin{corollary}
\label{cor:finitary-pred-transf}
  Let $T$ be finitary.
For each $f \colon X \to Y$ in $\Kl(T)$, the predicate transformer
  $\bbP^{\tau} (f) \colon \Omega^Y \to \Omega^X$
  is finitary.
\qed
\end{corollary}

\begin{theorem}[healthiness, in a finitary setting]
  \label{thm:finitary-healthiness}
  Let $T$ be a monad and
  $\bar{\Omega} =
    \p[\big]{\OmegaD, \bar{\tau} \colon T \to \cat{D}(\OmegaD^{(\place)}, \OmegaD)}$
  be a $\cat{D}$-relative $T$-algebra.
  Assume $T$ is finitary, and that $\bar{\tau}_X$ is surjective---much
 like
 in Theorem~\ref{thm:partialHealthiness}---but for each \emph{finite}
 set $X$. Then, for each map $\phi \colon \Omega^Y \to \Omega^X$,  the
 following are equivalent (healthiness).
 \begin{itemize}
  \item   There exists $f \colon X \to Y$ in $\Kl(T)$ such that
  $\bbP^{\tau}(f) = \phi$.
  \item
  $\phi$ is finitary (Definition~\ref{def:finitaryPredTrasn}) and lifts to $\cat{D}$. The latter means there
       exists $\bar{\phi} \colon \OmegaD^Y \to \OmegaD^X$
  such that $\bar{\phi} = \phi$. \qed
 \end{itemize}
\end{theorem}


\noindent
In particular, in case a monad $T$ is finitary, every
predicate transformer $\bbP^{\tau}(f) \colon \Omega^Y \to \Omega^X$ that arises
from a ``computation'' $f\colon X\to TY$ is finitary in the sense of
Definition~\ref{def:finitaryPredTrasn}. We will see that this is indeed the case for the
(sub)distribution monads (Sections~\ref{sub:subdist-monad-total}--\ref{sub:probabilistic-other}); these monads are finitary
because we restrict to (sub)distributions with a finite support.

We end with a topological interpretation of Definition~\ref{def:finitaryPredTrasn}.

\begin{proposition}
  \label{prop:finitary-continuous}
 Let $\Omega$ be a finite  set with the discrete topology.
 A predicate transformer $\phi \colon \Omega^Y \to \Omega^X$
  is finitary if and only if
  $\phi$ is  continuous with respect to the product topology of
 $\Omega^Y$ and $\Omega^X$.
 \qed
\end{proposition}

%
%

\section{(Purely) Probabilistic Examples}
\label{sec:prob-examples}

The term \emph{probabilistic computation} in the literature often refers
to one with an alternation of probabilistic and nondeterministic
branching, the latter modeling (totally unknown) environments'
behaviors, or a (demonic) \emph{scheduler}. This will be an example of our extended alternating
framework of Section~\ref{sec:two-player-setting}. Here we deal with
computations with purely probabilistic branching.

\subsection{Monads and Modalities for Probabilistic Branching
}
\label{subs:dist-monad}

One of the following monads replaces $\pow$ in Section~\ref{sec:nondet}.
We impose the restriction of countable supports.
\begin{definition}[the (sub)distribution monad $\dist,\sdist$]
The \emph{distribution monad} $\dist$ on $\Sets$ is such that:
 $\dist X=\{p:X\to [0,1]\mid \sum_{x\in X}p(x)= 1, \text{and $p(x)=0$
 for all but finitely many $x\in X$}\}$;
 $\dist f(p)(y)= \sum_{x\in f^{-1}(y)}p(x)$ on arrows; its unit is the
 Dirac distribution
 $\eta^{\dist}_{X}(x)(y)= 1$ (if $y=x$) and $0$ otherwise; and
 $\mu^{\dist}_{X}(\Phi)(x)=\sum_{p\in \dist X}\Phi(p)\cdot p(x)$.

 The \emph{subdistribution monad} $\sdist$
 is a variant defined by
 $\sdist X=\{p \text{ w/ finite supp.}
\mid \sum_{x\in X}p(x)\le 1\}$.
\end{definition}




 $\dist$-algebras are often called \emph{convex spaces}, with   convex
 subsets in  $\bbR^n$ as typical
 examples.
 $\dist$-algebra morphisms are  \emph{convex linear maps}, accordingly.
Since any (finite) convex combination can be expressed by
a repetition of
suitable binary convex combinations $x \oplus_p y = (1-p)x + py$ ,
a $\dist$-algebra structure is totally determined
by how  binary convex combinations are interpreted.

\begin{remark}
  Not all convex spaces are represented as convex subsets of
  $\bbR$-vector spaces: a two-point set $\sett{x, y}$ is
  a convex space by defining $(1-p)x + py$ as $x$ (if $p = 0$)
  and $y$ (otherwise). In general we have a monad map $\dist \to \fpow$
  to the finite powerset monad $\fpow$ that takes the \emph{support} of
 a distribution; consequently
 each join-semilattice (i.e.\ $\fpow$-algebras)
  yields a convex space.
  See~\cite{fritz2009convex} for more on
 convex spaces.
\end{remark}

A $\sdist$-algebra $x\colon \sdist X\to X$ is in turn called a
\emph{convex cone}, with the point $x(0)\in X$ (where $0$ is the \emph{zero}
subdistribution) assuming the special role of the \emph{apex} of a
cone. Indeed it is straightforward to see that a convex cone is a
``pointed convex space.''

\auxproof{
\begin{proposition}
  \label{prop:dist-sdist-algebra}
  There is a bijective correspondence between
  \begin{itemize}
    \item a $\sdist$-algebra structure on $X$, and
    \item a pair of a $\dist$-algebra structure on $X$ and
    a point of $X$ (called \emph{apex}).
  \end{itemize}
  Moreover a map $f \colon X \to Y$ between $\sdist$-algebras is
  a $\sdist$-algebra morphism if and only if
  it is $\dist$-algebra morphism and preserves the apex
  under the identification above.
\end{proposition}
\begin{proof}
  If we have a $\sdist$-algebra $e \colon \sdist X \to X$,
  then since $\dist \subset \sdist$ is a submonad,
  $e$ restricts to a $\dist$-algebra structure
  and designates a point of $X$ by $e(0)$.

  For the other direction, for a $\dist$-algebra structure
  $e' \colon \dist X \to X$ and a point $x \in X$,
  we will define $e \colon \sdist X \to X$ as
  $e(\mu) = e'(\mu + (1 - w(\mu)) \ket{x})$.
  Then we will check so defined $e$ is indeed a $\sdist$-algebra.
  Obviously we have $e(\ket{x}) = x$, so we only have to see
  \begin{align*}
    e \p*{\sum_{i} a_i \ket{e(\mu_i)}} = e \p*{\sum_{i} a_{i} \mu_i} \enspace.
  \end{align*}
  and it in fact follows from the EM-law for a $\dist$-algebra $e'$.

  We can easily check these operations are mutually inverse.
\end{proof}

Proposition~\ref{prop:dist-sdist-algebra} says that a convex cone is
a ``pointed convex space'', so we will call a $\sdist$-algebra morphism
a \emph{pointed convex linear map}. As in the case of convex spaces,
a convex cone structure is fully determined by the interpretation of
the apex and the binary convex combinations.
}

Let us turn to \emph{modalities}---i.e.\ $\dist$- and
$\sdist$-algebras---that would induce predicate transformer semantics by
the \emph{modality} recipe (Theorem~\ref{thm:theModalityRecipe}).
We adopt $\Omega=[0,1]$, the unit interval, with the intuition that
``probabilistic predicates'' are $[0,1]$-valued random
variables whose values express the likelihood of truth.

\begin{definition}[modalities for $\dist,\sdist$]
\label{def:probModalities}
 For $\dist$ we use
 \fussy
 $\tau\colon \dist[0,1]\to [0,1]$; it uses
 the usual convex structure of $[0,1]$.
 \sloppy

 For $\sdist$ we have a continuum of modalities:
 for each real number $r\in\uintv$ a modality
 $\tau_r \colon \sdist{\uintv} \to \uintv$
 is given by
 $\tau_r(p) = \sum_{x \in \uintv} x p(x) + r (1 - \sum_x p(x))$.

 We will in particular use the two extremes
$\tau_{\mathrm{total}} = \tau_0$
and $\tau_{\mathrm{partial}} = \tau_1$; they are called
 the \emph{total} and \emph{partial} modalities for $\sdist$, respectively.
 In the latter \emph{divergence}---whose probability is expressed by
$1 - \sum_x p(x)$---is deemed to yield truth. Hence
$\tau_{\mathrm{total}}$
and $\tau_{\mathrm{partial}}$ are analogues of $\Diamond$ and $\Box$ in
the nondeterministic setting (Section~\ref{sec:nondet}).
\end{definition}



\subsection{Healthiness for: Possibly Diverging Probabilistic
  Computations
and the Total Modality}
\label{sub:subdist-monad-total}

We shall first focus on the subdistribution monad $\sdist$ and the
total modality $\tau_{\mathrm{total}}$. These data give rise to
predicate transformer semantics---in the form of a state-and-effect
triangle~(\ref{eq:theModalityRecipe})---via the \emph{modality} recipe
(Theorem~\ref{thm:theModalityRecipe}). In particular we obtain
a functor $\bbP^{\tau_{\mathrm{total}}}\colon \Kl(\sdist)\to
\Sets^{\op}$ for interpreting a function $X\to \sdist Y$; the latter is
identified with
a \emph{probabilistic computation} from $X$ to $Y$ that is possibly
diverging (accounted for by \emph{sub}-probabilities).

Our goal is a healthiness result in this setting, towards which we rely on our \emph{relative algebra}
recipe. As we noted the original
Theorem~\ref{thm:state-effect-triangle-relative} is not enough; we use
its finitary variant (Theorem~\ref{thm:finitary-healthiness}), providing
its ingredient (a relative algebra $\bar{\Omega}$) by means of the
lifting result (Proposition~\ref{prop:lifting-condition}).

It turns out that the category $\cat{D}$ in the relative algebra recipe is
given by so-called \emph{generalized effect modules}. They have been
used in the context of categorical quantum logics~\cite{Jacobs15LMCS}
and the more general theory of effectuses~\cite{ChoJWW15Arxiv}.

\begin{definition}[$\GEMod$]
\label{def:gemod}
 A \emph{partial commutative monoid (PCM)} is a set $M$ with a partial binary
 \emph{sum} $\ovee$ and a \emph{zero} element $0 \in M$ that are subject
 to:
 $(x\ovee y) \ovee z \simeq x \ovee (y \ovee z)$,
 $x \ovee 0 \simeq x$ and $x \ovee y \simeq y \ovee x$,
 where $\simeq$ is the Kleene equality.
 A \emph{generalized effect algebra} is a PCM $(M, \ovee, 0)$
 that is  \emph{positive}
 ($x \ovee y = 0 \Rightarrow x = y = 0$)
 and \emph{cancellative} ($x \ovee y = x \ovee z \Rightarrow y = z$).

 A \emph{generalized effect module} is a generalized effect algebra $M$
 with a scalar multiplication ${\cdot} \colon [0, 1] \times M \to M$
 that satisfies
 $(r \ovee s) \cdot x \simeq (r \cdot x) \ovee (s \cdot x)$,
 $r \cdot (x \ovee y) \simeq (r \cdot x) \ovee (r \cdot y)$,
 $1 \cdot x = x$ and $r \cdot (s \cdot x) = (r \cdot s) \cdot x$.
 Here for $r, s \in [0, 1]$
 the partial sum $r \ovee s = r + s$ is defined when $r + s \leq 1$.

 The category of general effect modules (with a straightforward notion
 of
 their morphism, see~\cite{Cho15}) is denoted by $\GEMod$.
\end{definition}
\noindent
An example of a generalized effect module is the set $\sdist
X$ of subdistributions over $X$. Here $p\ovee q\in\sdist X$ is given
by $(p\ovee q)(x)=p(x)+q(x)$, which is well-defined clearly only if
 $\sum_{x\in X}p(x)+q(x)\le 1$. The set $\sdist X$ comes with an obvious
 scalar multiplication, too.
The unit interval $\uintv$ is another example; so is its products
$\uintv^{X}$.
See~\cite{Cho15} for details.



For our purpose of healthiness conditions,
we have to study
the monad map induced by the $\sdist$-algebra $\tauTotal$. It shall
also be denoted by $\tauTotal$. The following is easy.

\begin{lemma}\label{lem:monadMapFromTauTotal}
The monad map $\tauTotal\colon \sdist\to \Sets(\uintv^{(\place)},
 \uintv)$  is concretely given by:
  $(\tauTotal)_X(p)(f)
  = \sum_{x \in X} f(x) p(x)$.
\begin{enumerate}
 \item
 It lifts to
 $\tauTotal\colon \sdist\to \GEMod(\uintv^{(\place)},
 \uintv)$, that is, the map $(\tauTotal)_X(p)\colon [0,1]^{X}\to [0,1] $ for each $X$ and $p\in
 \sdist X$ preserves $0$, $\ovee$ and scalar multiplication.
 \item Furthermore
$(\tauTotal)_{Y}\colon \sdist Y\to \GEMod(\uintv^{Y},
 \uintv)$ is an isomorphism for each \emph{finite} set $Y$.
  \myqed
\end{enumerate}
\end{lemma}





\auxproof{
\begin{proof}
  Let us construct the inverse
  $\mu_X \colon \GEMod(\uintv^X, \uintv) \to \sdist X$.
  Given $\phi \colon \uintv^X \to \uintv$ in $\GEMod$,
  construct a subdistribution $\mu_X(\phi)$ on $X$ by
  $\mu_X(\phi) = \sum_{x \in X} \phi(\delta_x) \ket{x}$ where
  $\delta_x(x) = 1$ and $\delta_x(y) = 1$ for $y \neq x$.
  The sum is well-defined by the finiteness of $X$.
  We also have $\sum_{x \in X} \phi(\delta_x) \leq 1$ since
  $\sum_{x \in X} \delta_x \le 1$ and
  $\sum_{x \in X} \phi(\delta_x) = \phi \p*{\sum_{x \in X} \delta_x} \le 1$.
  It is easy to check $\sigma_X$ and $\mu_X$ are mutually inverse.
\end{proof}
}

From the last lemma the following healthiness result follows
immediately, via our general results. Specifically:
Lemma~\ref{lem:monadMapFromTauTotal}.1 discharges the condition of
Proposition~\ref{prop:lifting-condition} and provides the ingredient for
the relative algebra recipe; we then exploit
Lemma~\ref{lem:monadMapFromTauTotal}.2 and that $\sdist$ is finitary
in applying Theorem~\ref{thm:finitary-healthiness}.
\begin{theorem}[healthiness for $\sdist$ and $\tauTotal$]
 For a function $\varphi \colon \uintv^Y \to \uintv^X$ the following
 are equivalent: 1) there is $f\colon X\to \sdist Y$ such that
 $\varphi=\bbP^{\tauTotal}(f)$; 2) $\phi$ is finitary
 (Definition~\ref{def:finitaryPredTrasn}) and is a morphism of
 generalized effect modules, meaning that $0$, $\ovee$ and scalar
 multiplications are preserved by $\varphi$. \qed
\end{theorem}

\auxproof{
\begin{remark}
  The (sub)distribution monad we use is the \emph{finite-support}
  (sub)distribution monad; in \cite{Jacobs15CALCO}
  the \emph{countable-support}
  distribution monad can be represented
  using the dual adjunction
  between $\Set$ and $\DcEMod$, where the latter is the subcategory of
  $\EMod$ that additionally impose morphisms to be Scott-continuous.
\end{remark}
}

\subsection{Healthiness for Other Variations}
\label{sub:probabilistic-other}
For the other two variations of monads and modalities we can use the
same arguments as in Section~\ref{sub:subdist-monad-total}.

For the combination of $\sdist$ and the other modality $\tauPartial$,
we use the same category $\GEMod$ as an ingredient $\cat{D}$; the
difference is that the induced monad map $\tauPartial$ is ``dualized.''
\begin{theorem}[healthiness for $\sdist$ and $\tauPartial$]
  For a function
  $\varphi \colon \uintv^Y \to \uintv^X$ the following are equivalent: 1)
  there is  $f\colon X\to \sdist Y$ such that
  $\varphi = \bbP^{\tauPartial}(f)$; 2) $\phi$ is finitary and is a morphism of
  generalized effect modules. Here
  $[0, 1]$ is regarded as a generalized effect module
  in the  way dual to usual: its zero element is $1$, the partial sum $\owedge$ is defined
  by $x \owedge y = x + y - 1$ (if the right hand side is in $[0, 1]$)
  and  scalar multiplication $\co$ is defined by
  $r \co x = r \cdot x + (1 - r)$. \qed
\end{theorem}

For the (not sub-) distribution monad $\dist$ and the modality $\tau$
in Definition~\ref{def:probModalities}, we use the category $\EMod$
of \emph{effect modules} in place of $\GEMod$.
\begin{definition}[$\EMod$]\label{def:EMod}
 An \emph{effect module} $M$ is a generalized effect module that
 additionally has a \emph{top element} $1$. It is required to be the
 greatest with respect to  the canonical order $\le$ on $M$, defined by $x\le y$ if $y=x\ovee
 z$ for some $z\in M$.
 Effect modules and their morphisms---functions that preserve
 $0,1,\ovee$ and scalar multiplication---form a category denoted by
 $\EMod$.
\end{definition}

\begin{theorem}[healthiness for $\dist$ and $\tau$]
  For a predicate transformer function
  $\varphi \colon \uintv^Y \to \uintv^X$ the following are equivalent: 1)
  there is  $f\colon X\to \dist Y$ such that
  $\varphi = \bbP^{\tau}(f)$; 2) $\phi$ is finitary and is a morphism of
  effect modules, meaning that
  $0$, $1$, $\ovee$ and scalar
  multiplication are preserved by $\varphi$. \qed
\end{theorem}

\auxproof{
\subsection{Subdistribution Monad with Partial Modality}
\label{sub:subdist-monad-partial}

\subsection{Distribution Monad}
\label{sub:dist-monad}

As in the subdistribution case, we can define
a monad map $\sigma_X \colon \dist X \to \EMod(\uintv^X, \uintv)$
by integration, which is bijective when $X$ is finite.
The proposition below is proved similarly.

\begin{proposition}
  A predicate transformer $\phi \colon \uintv^Y \to \uintv^X$ comes from
  some $P \colon X \to \sdist Y$ if and only
  $\phi$ is finitary and an effect module morphism.
\end{proposition}

}

\section{Alternating Branching}
\label{sec:two-player-setting}
In this last section we  further extend our general framework to
accommodate \emph{alternating} branching, in which two players in
conflicting interests interplay. Its instances are  pervasive
in computer science, such as:  \emph{games}, i.e.\ a two-player variant
of automata, in which two players alternate in choosing next states (see
e.g.~\cite{Wilke01});
and various modeling of \emph{probabilistic systems} where it is
common to include additional nondeterministic branching for modeling
demonic behaviors of the environments (or \emph{schedulers}). See
e.g.~\cite{Sokolova05}.

In~\cite{Hasuo15TCS} the  \emph{modality} recipe
(Theorem~\ref{thm:theModalityRecipe}) is extended to alternating
branching; the central observation  is a compositional
treatment of two branching layers, using a monad $T$ on $\Sets$ (for
one)
and a monad $R$ on $\EM(T)$ (for the other). See~(\ref{eq:monad-composition}) later. It turns out that the same idea
works for our current generalized \emph{relative algebra} recipe
(Theorem~\ref{thm:state-effect-triangle-relative}).

After
the general framework we will describe some examples. A notable
one is (a variant of)  \emph{probabilistic predicate transformers}~\cite{MorganMS96}.

\subsection{The Relative Algebra Recipe for Alternation}
\label{sub:general-theory-two-player}


\begin{definition}[$R\star T$]
  Let $T$ be a monad on $\Set$ and $R$ be a monad on $\EM(T)$.
  Then a monad $R \star T$ is defined by the composite of the
  canonical adjunction $F \dashv U \colon \EM(T) \to \Set$
  and the monad $R$.
  \begin{equation}\label{eq:monad-composition}
   \small
    \vcenter{
    \xymatrix@R-1.2em{
      \Set
        \ar@(lu,ld)_-{R \star T = U R F}[]
        \ar@/^/[rr]^-{F}
        \ar@{}[rr]|-{\bot}
      &&  \EM(T) \ar@/^/[ll]^-{U}
        \ar@(ru,rd)^-{R}[]
    }}
  \end{equation}
\end{definition}
\begin{example}\label{ex:rStarTExample}
 One example is given by $T=\pow$ on $\Sets$ and $R=\Upx$, the
 \emph{up-closed powerset monad}, on $\EM(\pow)\cong \CL_{\biglor}$.
 It is given by
 \begin{math}
      \Upx (L, \le) = (\set{\calS \subset L}{\text{$\calS$ is up-closed}}, \supseteq)
 \end{math}; note that the inclusion order is reversed.
 This combination is for alternating branching where both layers are
 nondeterministic.

Another is given by
 $T=\dist$ on $\Sets$ and $R=\RC$, the
 \emph{nonempty convex powerset monad}, on $\EM(\dist)\cong \Conv$ (the
 category of convex spaces and convex-linear maps). The latter is given
 by
 \begin{multline*}
  \RC X =\{S\subset X\mid \text{
  $S$ is nonempty and \emph{convex-closed}, i.e.\
}
  \\
  \text{  $x_{1},\dotsc,x_{n}\in S$ and $\lambda_{1}+\cdots+\lambda_{n}=1$
  implies $\textstyle\sum_{i}\lambda_{i}x_{i}\in S$}\}\enspace.
 \end{multline*}
 This is
 for  alternating probabilistic and nondeterministic branching.

 It is important that these $R$ are not monads on $\Sets$ per se; they involve
 $T$-algebra structures.
\end{example}

\begin{remark}
  We have comparison functors
  $K \colon \Kl(R \star T) \to \Kl(R)$ and
  $L \colon \EM(R) \to \EM(R \star T)$ as follows.
  \begin{equation}\label{eq:two-player-comparison}
   \small
    \vcenter{
    \xymatrix@R-2em@C-2em{
      \Kl(R \star T) \ar[r]^-{K}
      & \Kl(R) \ar@/^/[rd]
      && \EM(R) \ar@/^/[ld] \ar[r]^-{L}
      & \EM(R \star T) \\
      && \EM(T) \ar@/^/[lu] \ar@/^/[ru] \ar@/^/[d] \ar@{}|{\dashv}[lu]
          \ar@{}|{\dashv}[ru]\\
      && \Set \ar@/^/[u] \ar@(lu, ld)_{R \star T} \ar@{}|{\dashv}[u]
    }}
  \end{equation}
\end{remark}

We aim at reproducing the \emph{relative algebra} recipe
(Theorem~\ref{thm:state-effect-triangle-relative}) for  the current
alternating setting. The first ingredient for the recipe was
a dual adjunction
\begin{math}
  \xymatrix@1@C-.5em{
    {\Sets}
      \ar@/^.4em/[r]
      \ar@{}[r]|-{\bot}
    &
    {\cat{D}^{\op}}
      \ar@/^.4em/[l]
  }
\end{math}
from which we derived a continuation-like monad
$\cat{D}(\OmegaD^{(\place)},\OmegaD)$. For the alternating version of the recipe,
in view of~(\ref{eq:monad-composition})
it is
natural to use
a dual adjunction
\begin{math}
  \xymatrix@1@C-.5em{
    {\EM(T)}
      \ar@/^.4em/[r]
      \ar@{}[r]|-{\bot}
    &
    {\cat{D}^{\op}}
      \ar@/^.4em/[l]
  }
\end{math}
(where $\EM(T)$ replaced $\Sets$ in the non-alternating counterpart).
Interestingly,
for such an ingredient
\begin{math}
  \xymatrix@1@C-.5em{
    {\EM(T)}
      \ar@/^.4em/[r]
      \ar@{}[r]|-{\bot}
    &
    {\cat{D}^{\op}}
      \ar@/^.4em/[l]
  }
\end{math}
we
can exploit   (the original, non-alternating version of) relative algebra
recipe itself. See~(\ref{eq:state-effect-triagle-relative}) in
Theorem~\ref{thm:state-effect-triangle-relative}; this yields a
continuation-like monad $[[\place , \bar{\Omega}]_{T} ,
\bar{\Omega}]_{\cat{D}}$ over $\EM(T)$.

Then it is clear that the next key ingredient in the original
recipe---namely a monad map
$\tau \colon T \to \cat{D}(\OmegaD^{(\place)}, \OmegaD)$---has
 a monad map
$\rho \colon R \to [[\place , \bar{\Omega}]_{T} ,
\bar{\Omega}]_{\cat{D}}$ as its alternating counterpart.

Given such data we obtain predicate transformer semantics.

\begin{definition}[$\bbP^{\rho}$, $\bbP^{(\tau, \rho)}$]
  \label{def:alternating-pred-transf}
  Let $T$ be a monad on $\Set$, $\cat{D}$ be a complete category
  and $\bar{\Omega} = (\OmegaD, \tau)$ be a $\cat{D}$-relative $T$-algebra.
  Moreover let $R$ be a monad on $\EM(T)$
  and $\rho \colon R \to [[\place , \bar{\Omega}]_{T} , \bar{\Omega}]_{\cat{D}}$ be a monad map.
  Here $[[\place , \bar{\Omega}]_{T} , \bar{\Omega}]_{\cat{D}}$ is the monad
  that arises from the dual adjunction
  \begin{math}
    \xymatrix@1@C-.5em{
      {\EM(T)}
        \ar@/^.4em/[r]
        \ar@{}[r]|-{\bot}
      &
      {\cat{D}^{\op}}
        \ar@/^.4em/[l]
    }
  \end{math}
  induced by $\bar{\Omega}$ as in Theorem~\ref{thm:state-effect-triangle-relative}.

 We define a functor $\bbP^{\rho} \colon \Kl(R) \to \cat{D}^{\op}$
 by:
  $\bbP^{\rho} A_a = A_a$ and
  \begin{align*}
    \bbP^{\rho}(A_a \stackrel{f}{\to} B_b
 \text{ in $\Kl(R)$}
)
      &= \bigl(\,[B_b, \bar{\Omega}]_{T} \stackrel{\rho^{\sharp}}{\to}
          [R B_b, \bar{\Omega}]_{T} \stackrel{f^*}{\to}
            [A_a, \bar{\Omega}]_{T} \,\bigr)\enspace .
  \end{align*}
 Furthermore,
by  precomposing
 the comparison functor $K \colon \Kl(R \star T) \to \Kl(R)$,
  we have another functor:
  \begin{equation}
    \bbP^{(\tau, \rho)} = \bbP^{\rho} \co K \;\colon\; \Kl(R \star T) \longto \cat{D}^{\op}
    \enspace .
  \end{equation}
\end{definition}
\noindent
The last functor $\bbP^{(\tau, \rho)}$ is what we want: it interprets
a function $X\to (R\star T)Y$---a computation from $X$ to $Y$, with
alternation of $T$- and $R$-branching---in the category
$\cat{D}$, in a backward manner.

 The following extends Theorem~\ref{thm:partialHealthiness}.

\begin{theorem}[alternating healthiness condition]
  \label{thm:alternating-healthiness}
 Assume the setting of Definition~\ref{def:alternating-pred-transf}, and
 let $X$ and $Y$ be sets. If the map
  \begin{align*}
    U \rho_{FY} \colon U R F Y
        \longto U [[FY, \bar{\Omega}]_{T}, \bar{\Omega}]_{\cat{D}}
        \cong \cat{D}(\OmegaD^Y, \OmegaD)
  \end{align*}
  is surjective (injective), then the action
  $\bbP^{(\tau, \rho)}_{XY} \colon \Kl(R \star T)(X, (R \star T) Y) \to
    \cat{D}(\Omega^Y, \Omega^X)$ of $\bbP^{(\tau, \rho)}$ is surjective
 (injective). \qed
\end{theorem}

Let us now assume that $\cat{D}$ is concrete, and develop an alternating
counterpart of Section~\ref{subsec:relative-alg-concrete}.


\auxproof{
In the relative algebra recipe,
we need not to define a monad map
$\bar{\tau} \colon T \to \cat{D}(\OmegaD^{(\place)}, \OmegaD)$ directly.
It is enough to give a Eilenberg-Moore $T$-algebra
$\hat{\tau} \colon T{\Omega} \to \Omega$
and check the lifting condition;
moreover we only have to check the condition \emph{pointwise}
as stated in Proposition~\ref{prop:lifting-condition}.
Similarly in alternating setting,
to acquire a monad map
$\bar{\rho} \colon R \to [[\place , \bar{\Omega}]_{T} , \bar{\Omega}]_{\cat{D}}$
it is enough to give an $R$-algebra
$\hat{\rho} \colon R{\Omega_{\tau}} \to {\Omega_{\tau}}$
that is ``compatible''.
Recall that for a $T$-algebra $A_a$,
$[[A_a , \bar{\Omega}]_{T} , \bar{\Omega}]_{\cat{D}}$
is a $T$-subalgebra of
$\Omega_{\tau}^{V [A_a , \bar{\Omega}]_{T}}
  \cong \Omega_{\tau}^{\EM(T)(A_a, \Omega_{\tau})}$
by Proposition~\ref{prop:state-pred-duality-forgetful}.
}
\begin{theorem}
  \label{thm:two-player-lifting}
  Suppose we have a monad $T$ on $\Set$, a complete category $\cat{D}$
 with a  faithful and limit-preserving functor
 $V \colon \cat{D}
\to \Set$,
and
  a $\cat{D}$-relative $T$-algebra $\bar{\Omega} = (\OmegaD, \bar{\tau})$.
  Moreover assume we have a monad $R$ on $\EM(T)$ and
  an $R$-algebra structure $\hat{\rho} \colon R{\Omega_{\tau}} \to {\Omega_{\tau}}$
  on ${\Omega_{\tau}}$. Then the following are equivalent.
  \begin{enumerate}
    \item The monad map $\rho \colon R \to {\Omega_{\tau}}^{\EM(T)(\place, \Omega_{\tau})}$
      that corresponds to $\hat{\rho}$
      (Proposition~\ref{prop:emalg-monadmap})
      lifts to a monad map $\bar{\rho} \colon R \to [[\place ,
	  \bar{\Omega}]_{T}, \bar{\Omega}]_{\cat{D}}$, equipped with a
	  suitable $\cat{D}$-structure.
    \item \emph{(Lifting condition)} For each $T$-algebra $A_a$, the extension map
      $\rho_{A_a}^{\sharp} \colon \EM(T)(A_a, \Omega_{\tau}) \to \EM(T)(R A_a, \Omega_{\tau})$,
      which maps a $T$-algebra morphism $f \colon A_a \to \Omega_{\tau}$
      to $\hat{\rho} \co Rf \colon R A_a \to \Omega_{\tau}$,
      lifts (along $V$) to a $\cat{D}$-morphism
      $\bar{\rho}^{\sharp} \colon [A_a, \bar{\Omega}]_{T} \to [R A_a, \bar{\Omega}]_{T}$.
      That is, there exists $\bar{\rho}^{\sharp}$ such that $V \bar{\rho}^{\sharp} = \rho^{\sharp}$.
    \item \emph{(Pointwise lifting condition)} For each $A_a$ and $x \in U R A_a$,
    the map
    $({\rho}^{\sharp})_x = \pi_x \after \rho^{\sharp} \colon \EM(A_a, \Omega_{\tau}) \to \Omega$
    lifts to a $\cat{D}$-morphism
    $(\bar{\rho}^{\sharp})_x \colon [A_a, \bar{\Omega}] \to \OmegaD$.
    Here $\pi_x$ is defined by the following composite:
    \begin{equation*}
      \pi_x =\bigl(\,
        \EM(T)(R A_a, \Omega_{\tau}) \xxtos{U} \Set(U R A_a, \Omega)
          \xxtos{\ev_x} \Omega \,\bigr) \enspace \text{in $\Set$}
    \end{equation*}
  and $\ev_x$ evaluates a function $h \colon U R A_a \to \Omega$ by $x$.
  \qed
  \end{enumerate}
\end{theorem}
\noindent  In the above we implicitly used the isomorphism
  $V [A_a, \bar{\Omega}]_{T} \cong \EM(T)(A_a, \bar{\Omega})$
  given in Proposition~\ref{prop:state-pred-duality-forgetful}.

\subsection{Examples}\label{sub:alternatingExamples}
We list  healthiness results for some alternating situations.
  We indicate how we exploit the
general framework above; the details are omitted for space reasons.

\paragraph{Nondeterminism and Divergence}
In Dijkstra's original work~\cite{Dijkstra76} the first healthiness
result is presented for computations with alternation between
\emph{divergence} and \emph{nondeterminism}. They are described by
functions of the type $X\to (\pow_{+}\star\lift) Y$, where: $\lift X=1+X$ is
the \emph{lift} monad on $\Sets$ (modeling potential divergence); and
$\pow_{+}$ is the \emph{nonempty powerset} monad on
$\EM(\lift)\cong\Sets_{*}$, the category of \emph{pointed sets}. The
latter monad is given specifically by $\pow_{+}(X,x)=\bigl(\{S\subseteq
X\mid S\neq \emptyset\}, \{x\}\bigr)$.

Suitable modalities $\tau$ and
$\rho$ are found to capture the setting of~\cite{Dijkstra76}.
For the category $\cat{D}$ for predicate transformers we introduce
the notion of \emph{strict complete meet-semilattice}. It is a poset
with
the least element $0$ and
arbitrary but nonempty meets.

\begin{theorem}[healthiness for  nondeterminism and divergence]
  For a function
  $\varphi \colon 2^Y \to 2^X$ the following are equivalent: 1)
  there is  $f\colon X \to (\pow_+ \star\lift) Y$ such that
  $\varphi = \bbP^{(\rho, \tau)}(f)$; 2) $\phi$
  preserves $0$ and nonempty meets.
  \qed
\end{theorem}

\paragraph{Alternating Nondeterminism}
Alternation of two layers of nondeterminism is found e.g.\ in \emph{games}. In
program logic point of view---one player  \emph{ensures} a
postcondition, no matter the other player's move is---such computation
is best modeled as a function $X\to (\Upx\star\pow) Y$. Here $\Upx$ is the monad
on $\EM(\pow)\cong \CL_{\biglor}$ from Example~\ref{ex:rStarTExample}.
There are  modalities $\tau$ and
$\rho$ suited to capture the above game-theoretic intuitions;
see~\cite{Hasuo15TCS}. Here we choose a combination in which: the
opponent moves first, and the protagonist follows.
Towards healthiness we take posets and monotone functions in
$\cat{D}$.
\begin{theorem}[healthiness for alternating nondeterminism]
  For a function
  $\varphi \colon 2^Y \to 2^X$ the following are equivalent:
  1) there is  $f \colon X \to (\Upx\star\pow) Y$ s.t.\
  $\varphi = \bbP^{(\rho, \tau)}(f)$; 2) $\phi$ is monotone. \qed
\end{theorem}

\paragraph{Nondeterminism and Probability}
Finally we study a common setting in the study of probabilistic systems
where: a demonic nondeterministic choice occurs first, followed by an
angelic probabilistic choice. This is modeled by a function $X\to
(\RC\star \dist)Y$, where $\RC$ on $\EM(\dist)\cong\Conv$ is from
Example~\ref{ex:rStarTExample}.
Predicate transformer semantics of such computations has $[0,1]$ as the
domain of truth values; and suitable modalities $\tau$ and
$\rho$ are found much like in~\cite{Hasuo15TCS}. The outcome is (a
slight variation of) \emph{probabilistic predicate transformers} in~\cite{MorganMS96}.

For healthiness we use the category
of: effect algebras (Definition~\ref{def:EMod}) and what we call
\emph{regular-sublinear maps} between them. The latter are subject to:
(\emph{subadditivity}) if $x \perp y$ then $f(x) \perp f(y)$ and we
have $f(x) \ovee f(y) \leq f(x \ovee y)$; (\emph{scaling}) $f(\lambda
x) = \lambda f(x)$; and (\emph{translation}) $f(x \ovee \lambda 1) =
f(x) \ovee \lambda 1$ if $x \perp \lambda 1$. It deviates from
\emph{sublinear maps}~\cite{MorganMS96} in that we require $=$ in
(translation).


In the following we assume $Y$'s finiteness; this is like in~\cite{MorganMS96}.

\begin{theorem}[healthiness for nondeterminism and probability]
  Assume $Y$ is finite.
  For a function
  $\varphi \colon [0, 1]^Y \to [0, 1]^X$ the following are equivalent:
  1) there is  $f \colon X \to (\RC\star \dist) Y$ such that
  $\varphi = \bbP^{(\rho, \tau)}(f)$; 2) $\phi$ is regular-sublinear.
 \qed
\end{theorem}

\auxproof{

\subsection{Nondeterminism with Failure}
\label{sub:nondet-failure}

In Dijkstra's original paper~\cite{Dijkstra76} where the healthiness condition
is first introduced, they essentially models a nondeterministic
program as a map $f \colon X \to \pow_+(Y+1)$
where $\pow_+(X)$ is the set of non-empty subsets of $X$.
In the two-player PT situation framework, this situation can be understood
as the case of $T = \calL$ and $R = \pow_+ \colon \Sets_* \to \Sets_*$,
where $\calL$ is the ``lift'' monad $\calL = (\place) + 1$ on $\Sets$
and $\Sets_*$ is the category of pointed sets.
\memo{Ref: \cite{morgan1996unifying}}

According to \cite{morgan1996unifying}, the weakest precondition
for a computation $f \colon X \to \pow_+(Y + 1)$ is
$\wpre(f) \colon 2^{Y + 1} \to 2^{X + 1}$, defined by
\begin{align*}
 wp(f)(\pi)(x) = \bigland_{y \in r(x)} y \enspace .
\end{align*}
The healthiness condition for $p \colon 2^{Y + 1} \to 2^{X + 1}$ is
\begin{itemize}
 \item strictness: $p(\mathbf 0) = \mathbf 0$, where $\mathbf 0$ is the constant predicate that returns 0, and
 \item positive conjunctivity: for every nonempty set $I$, $p(\bigland_{i \in I}\pi_i) = \bigland_{i\in I}p(\pi_i)$.
\end{itemize}
We introduce a new category $\StrCL$ of strict complete meet-semilattices, where objects have 0 and positive meets, and arrows preserve 0 and positive meets.
We have adjunctions
\begin{align*}
 \xymatrix{
 \Sets \ar@/^/[rr]^{F_\calL = (\place) + 1} && \Sets_* = \EM(\calL)\ar@/^/[ll]^{U_\calL = (\place)} \ar@/^/[rr]^{2^{(\place)}}&&  \StrCL^\op \ar@/^/[ll]^{\StrCL(\place,2)}
 }
\end{align*}
where the left adjunction yields $\calL$ and the right adjunction the ``powerset-like'' monad $(X, x) \mapsto \pow_+(X \setminus \{x\}) + 1$ on $\Sets_*$.
The composition of them yield the monad $\pow_+(\place) + 1$.
There is a (pointwise) surjective monad map $\alpha \colon \pow_+ \star \calL = \pow_+(\place + 1) \To \pow_+(\place) + 1$, defined by
\begin{align*}
 \alpha_X \colon \pow_+(X + 1) \longto \pow_+(X) + 1 ; \quad V \mapsto
 \begin{cases}
  \bot & (\bot \in V)\\
  V & (\mbox{otherwise})
 \end{cases}
 \enspace .
\end{align*}
Therefore, the healthiness condition is trivial in this situation: a predicate transformer $p \colon 2^Y \to 2^X$ satisfies healthiness condition if and only if $p$ is an arrow in $\StrCL$.

\subsection{Nondeterministic-Nondeterministic Programs}
\label{sub:nondet-nondet}

\begin{notation}
  When we want to make explicit the distinction between
  subsets and characteristic functions,
  we write $D_f = \set{x \in X}{f(x) = 1}$
  for a characteristic function $f \colon X \to 2$ on $X$,
  and $\chi_f = \lambda x \in X \ldotp (x \in S)$ for $S \subset X$.
\end{notation}

\begin{definition}
  We define two monads; the one is $\UP$ defined on $\Sets$ and
  the other is $\Upx$ on $\CL_{\biglor}$.
  \begin{align*}
    \UP X &:= \set{\calS \subset \pow X}{\text{$\calS$ is up-closed}} \\
    \Upx (L, \le) &:= (\set{\calS \subset \pow L}{\text{$\calS$ is up-closed}}, \supseteq)
  \end{align*}
  Note that in $\Upx$, the order is the \emph{reverse} inclusion order
  rather than the usual inclusion order.
  More intuitively, we define the order by $L \le L'$ if and only if
  for each $x \in L$ there is some $y \in L'$ such that $x \le y$;
  it is equivalent to the reverse inclusion order.

  The unit and the multiplication are defined by
  \[
    \eta_X \colon X \longto \Upx X; x \longmapsto \upcl \{x\}
  \] 
  \cont
  For $N \colon X \to \UP Y$, we define a corresponding
  predicate transformer $\phi \colon 2^Y \to 2^X$ as
  $\phi(N) = \lambda x \ldotp \biglor_{S \in N(x)} \p*{ \bigland_{y \in S} f(y)}
    = \lambda x \ldotp \p*{D_f \in N(x)}$.
\end{definition}

\begin{remark}
  The monad $\UP$ on $\Set$ can be decomposed into the (monadic) adjunction
  $F \dashv U \colon \EM(\pow) \cong \CL_{\biglor} \to \Set$ and
  the up-closed powerset monad $\Upx$ on $\CL_{\biglor}$;
  that is, $\UP = U \co \Upx \co F$.
\end{remark}

We have the adjunction
$[\place, 2]_{\le} \dashv [\place, 2]_{\biglor} \colon \CL_{\biglor} \to \Pos^{\op}$
both endowed with the natural orders.

\begin{proposition}
  \label{prop:SPisomUpx}
  The associated monad $[[\place, 2]_{\biglor}, 2]_{\leq}$ is
  isomorphic to $\Upx$ by
  \begin{align*}
    \sigma_L \colon \Upx L &\longto [[L, 2]_{\biglor}, 2]_{\leq};
      U \longmapsto \lambda f \ldotp \bigland_{x \in U} f(x) \enspace .
  \end{align*}
\end{proposition}

\begin{proposition}
  A predicate transformer $\phi \colon 2^Y \to 2^X$ comes from
  some $N \colon X \to \UP Y$ if and only if it is monotone.
\end{proposition}
\begin{proof}

\end{proof}

\subsection{Nondeterministic-Probabilistic Programs}
\label{sub:nondet-prob}

The predicate transformer semantics for a program with nondeterministic
and probabilistic branching and its healthiness condition
are introduced in \cite{MorganMS96}.

In~\cite{Hasuo14, Hasuo15TCS}, we used the \emph{convex powerset monad}
on the category of \emph{convex cones}.

We use the \emph{convex powerset monad} on the category
of \emph{convex spaces} for simplicity.

\begin{definition}
  We define a \emph{nonempty convex powerset monad} $\RC$ on
  $\EM(D) \cong \Conv$ as follows.
  \begin{itemize}
    \item \emph{object}: for a convex space $X$,
      define $\RC X$ as the set of nonempty and convex subsets in $X$.
      The convex combination is
      $\sum_{i} \lambda_i C_i = \set{\sum_{i} \lambda_i x_i}{x_i \in C_i}$.
    \item \emph{morphism}: for a convex linear map $f \colon X \to Y$,
      a map $\RC f$ is defined by $\RC f (C) = f[C]$, where $f[C]$ is
      the image of $C$ by $f$.
    \item \emph{monad structure}: the unit $\eta \colon X \to \RC X$
    and the multiplication $\mu \colon \RC (\RC X) \to \RC X$ are defined by
    $\eta_X (x) = \sett{x}$ and $\mu_X (\calS) = \bigcup \calS$.
  \end{itemize}
  We can easily check the well-definedness of this definition.
\end{definition}

\begin{definition}
  A map $f \colon M \to N$ between effect modules is called \emph{subsublinear}
  \memo{need name} if the following conditions hold:
  \begin{itemize}
    \item \emph{subadditivity}: if $x \perp y$ then $f(x) \perp f(y)$ and
    we have an inequality $f(x) \ovee f(y) \leq f(x \ovee y)$,
    \item \emph{scaling}: $f(cx) = c f(x)$, and
    \item \emph{translation}: $f(x \ovee \lambda 1) = f(x) \ovee \lambda 1$
      when $x \perp \lambda 1$.
  \end{itemize}
  We denote by $\EMod_{\leq}$ the category of effect modules and
  subsublinear maps. Note that it contains $\EMod$ as a subcategory.
\end{definition}

\begin{proposition}
  The binary convex combination $\oplus_p \colon \uintv^2 \to \uintv$
  is subsublinear. Therefore it induces a dual adjunction between
  $\Conv$ and $\EMod_{\le}$; concretely
  $[\place, \uintv] \colon \Conv \to \EMod_{\le}$
  and $[\place, \uintv] \colon \EMod_{\le} \to \Conv$,
  where effect module structures and convex space structures
  are given pointwise.
\end{proposition}

\begin{definition}
  Define a $\RC$-algebra structure
  $\hat{\rho} \colon \RC \uintv \to \uintv$ on $\uintv$
  by $\hat{\rho}(C) = \inf C$.
  It clearly preserves convex combinations, hence is a convex linear map.
\end{definition}

\begin{proposition}
  The $\RC$-algebra $\hat{\rho} \colon \RC \uintv \to \uintv$
  satisfies the lifting condition over $\EMod_{\leq}$;
  for a convex space $X$, the extension map
  $\rho^{\sharp}_X \colon \Conv(X, \uintv) \to \Conv(\RC X, \uintv)$
  is subsublinear.
\end{proposition}
\begin{proof}
  For a convex set $C \subset \dist X$, the $C$-component
  $(\rho_X^{\sharp})_C$ is calculated by
  $(\rho_X^{\sharp})_C (f) = \inf_{x \in C} f(x)$.
  Then subadditivity, scaling and translation follow obviously.
\end{proof}

Then we have a monad map $\tau \colon \RC \to [[\place , \uintv], \uintv]$.

\begin{proposition}
  For the free convex space $\dist n \cong \Delta^n$ over a finite set $n$,
  we have that $\tau_{\Delta^n}$ is surjecive.
\end{proposition}
\begin{proof}
  Assume $\phi \colon \uintv^n \to \uintv$ is subsublinear.
  Define $C \subset \Delta^n$ by
  $C = \set{x \in \Delta^n}
        {\forall f \in \uintv^n \ldotp \phi(f) \leq f^{\sharp}(x)}$.
  It is easy to check $C$ is convex.
  We will prove $\phi(f) = \inf_{x \in C} f^{\sharp}(x)$.
  Since $\phi(f) \leq \inf_{x \in C} f^{\sharp}(x)$ directly follows from
  the definition of $C$, we only have to show in fact the equality holds.

  Assume $\phi(f) \lt \inf_{x \in C} f^{\sharp}(x)$ for some $f$.
  Then we have $\phi(f) \lt f^{\sharp}(x)$ for each $x \in C$.
  By the definition of $C$, it means $\p*{\bigcap_{g} F_g} \cap F'$
  is empty where
  $F_g = \set{x \in \Delta^n}{ \phi(g) \leq g^{\sharp}(x) }$
  for $g \in \uintv^n$ and
  $F' = \set{x \in \Delta^n}{ \phi(f) \geq f^{\sharp}(x) }$.
  Since $F_g$ and $F'$ is closed and $\Delta^n$ is compact,
  we have a finite number of $g_0, \ldots, g_k$ such that
  $\p*{\bigcap_{i \le k} F_{g_i}} \cap F'$ is also empty.
  By Farkas' lemma \memo{need cite?}, there are nonnegative
  $c_1, \ldots, c_k, d, \lambda, \lambda'$ such that
  \begin{align*}
    c_1 g_1 + \cdots + c_k g_k + \lambda 1 &= d f + \lambda' 1 \le 1 \enspace \text{, and} \\
    c_1 \phi(g_1) + \cdots + c_k \phi(g_k) + \lambda 1 &\gt d \phi(f) + \lambda' 1
    \enspace ,
  \end{align*}
  which contradicts the subsublinearity of $\phi$.
\end{proof}

\begin{proposition}
  Assume $Y$ is finite.
  A predicate transformer $\phi \colon \uintv^Y \to \uintv^X$ is described
  as $\phi(p) = \lambda x \ldotp \inf_{\mu \in f(x)} \int p \mu$
  using some state transformer $f \colon X \to \RC (\dist Y)$
  if and only if $\phi$ is subsublinear.

\end{proposition}
}

\acks
We thank Toshiki Kataoka for helpful discussions, and the anonymous
referees for useful comments. Special thanks are due
to John Power for the lectures he gave on the occasion of his visit to
Tokyo; the notion of relative algebra is inspired by them.
W.H., H.K.\ and I.H.\ are
 supported by Grants-in-Aid No. 24680001, 15K11984 \& 15KT0012,
 JSPS.


\bibliographystyle{abbrvnat}


\bibliography{./ref.bib}

%
%

\newpage
\appendix
\section{Omitted Proofs and Details}
\subsection{Explicit Definition of Monad Maps}
\label{appendix:monadMap}
\begin{definition}
 Let $S,T$ be monads on $\cat{C}$. A
 \emph{monad map} from $S$ to $T$ is
 a natural transformation $\alpha\colon S\to T$ that makes the following
 diagram commute.
 \begin{equation}
     \xymatrix@R-1.5em@C-1em{
        X \ar[r]^{\eta^S_X} \ar[rd]_{\eta^T_X}
        & SX \ar[d]^{\alpha_X} \\
        & TX
     }
     \xymatrix@R-1.5em@C-1.5em{
        SSX \ar[r]^{S \alpha_{X}} \ar[d]^{\mu^S_X}
        & STX \ar[r]^{\alpha_{TX}}
        & TTX \ar[d]_{\mu^T_X} \\
        SX \ar[rr]^{\alpha_{X}}
        && TX
   }
 \end{equation}
 Here $\eta^{(\place)}$ and $\mu^{(\place)}$ are the unit and the multiplication of monads.
\end{definition}

\subsection{Proof of Lemma~\ref{lem:monad-isom-jslat}}
\begin{proof}
  Let the monad $[2^{(\place)}, 2]_{\biglor}$ be denoted by $T$ in the
  current proof, for brevity.

  We need to check that  $\sigma_X$ is join-preserving for each
 $X$. Indeed, for each $S\subseteq X$, we have
  $\bigvee_{x \in S} \p[\big]{\bigvee_{f \in \calF} f(x)}
  = \bigvee_{f \in \calF} \p[\big]{\bigvee_{x \in S} f(x)}$
  for any family $\calF \subset 2^X$.
  It is easy to check that $\sigma$ is natural, and that it is
 compatible with monad units.
 Compatibility with monad multiplications requires the following diagram
 to commute.
  \begin{align*}
    \xymatrix@R=.8em{
      \pow \pow X \ar[r]^{\pow \sigma_X} \ar[d]^{\bigcup}
      & \pow T X \ar[r]^{\sigma_{T X}}
      & T T X \ar[d]^{\mu_X} \\
      \pow X \ar[rr]^{\sigma_X}
      && T X
    }
  \end{align*}
 Indeed, for each  $\calS \in \pow \pow X$ we have
  \begin{align*}
    \mu \circ \sigma_{T X} \circ \pow \sigma_X (\calS)
    &= \mu \circ \sigma_{T X} \p*{\set{\sigma_X (S)}{S \in \calS}} \\
    &= \mu \p*{\lambda \phi \ldotp
      \biglor_{S \in \calS} \phi \circ \sigma_X (S)} \\
    &= \lambda f' \ldotp \p*{\lambda \phi \ldotp
      \biglor_{S \in \calS} \phi \circ \sigma_X (S)} \p*{\lambda \xi \ldotp \xi(f')} \\
      &= \lambda f' \ldotp \biglor_{S \in \calS} \sigma_X (S)(f') \\
    &= \lambda f' \ldotp \biglor_{S \in \calS} \biglor_{x \in S} f'(x) \\
    &= \sigma_X \p*{\bigcup \calS} \enspace .
  \end{align*}
  Finally we check that
  $\sigma_{X}$ is bijective. Its inverse is  given by
  \begin{align*}
    (\sigma_X)^{-1}(\xi) = \set{x \in X}{\xi(\delta_x) = 1}\enspace,
  \end{align*}
  where $\delta_{x}\colon X\to 2$ is given by: $\delta_x(x) = 1$, and $\delta_x(y) = 0$ if $x \neq y$.
\end{proof}

\subsection{Proof of Theorem~\ref{thm:theModalityRecipe}}
\begin{proof}
  For any set $X$ we have $[KX, \Omega_{\tau}]_T \cong \Omega^X$
  since $KX$ is the free $T$-algebra over $X$.
  It is natural in $X \in \Kl(T)$. Indeed, for $f \colon X \to TY$,
  the  diagram
  \begin{align*}
    \xymatrix@R-1.5em{
      [TX, \Omega_{\tau}]_{T} \ar[d]^{\cong}
      && [TY, \Omega_{\tau}]_{T} \ar[ll]_{(Kf)^*} \ar[d]^{\cong} \\
      \Omega^X
      & \Omega^{TY} \ar[l]_{f^*}
      & \Omega^{Y} \ar[l]_{\tau^{\sharp}}
    }
  \end{align*}
commutes by direct calculation.
\end{proof}
\subsection{Proof of Proposition~\ref{prop:change-of-base}}
\begin{proof}
  Since the functor $H$ is product preserving,
  the canonical map $\theta \colon H(A^X) \to (HA)^X$ in
  $\cat{D'}$, defined by the transpose of
  \begin{equation}\label{eq:change-of-base}
    X \xxto{\id^{\sharp}} \cat{D}(A^X, A)
      \xxto{H} \cat{D'}(H(A^X), HA) \enspace ,
  \end{equation}
 is an isomorphism.
  Using this isomorphism $\theta$,
  we define a natural transformation $\psi^A$ by
  $\psi^A_X = (\theta^{-1})^* \co H$, that is,
  \begin{align*}
     \cat{D}(A^{X}, A)
      \xxto{H} \cat{D'}(H(A^{X}), HA)
      \xxto{(\theta^{-1})^*} \cat{D'}((HA)^{X}, HA)
   \enspace.
  \end{align*}
  This $\psi$ is  seen to be
  a monad map by some diagram chasing.
  We define a functor $\bar{H}$ by
  \begin{align*}
    \bar{H} \p*{A, \alpha} &= \p*{HA, \psi \after \alpha}
        \enspace, \quad\text{and} \\
    \bar{H} ((A, \alpha) \xxto{f} (B, \beta)) &=
        ((HA, \psi \after \alpha) \xxto{Hf} (HB, \psi \after \beta))
        \enspace .
  \end{align*}
  It is a routine to check that
$Hf$ is indeed a morphism of $\cat{D'}$-relative $T$-algebra,
  and that $\bar{H}$ makes the diagram in
 (\ref{eq:change-of-base-lifting}) commute.

 That $\bar{H}$ preserves products is easily checked by direct
 calculations. That $\bar{H}$ is faithful, given that $H$ is so, follows
 immediately from~(\ref{eq:change-of-base-lifting}).
\end{proof}

\subsection{Proof of Theorem~\ref{thm:state-effect-triangle-relative}}
\begin{proof}
  We will denote
  $\OmegaD^X$ and $\cat{D}(M, \OmegaD)$ by $X^*$ and $M^*$ respectively
  in this proof.

  We check the adjointness of $[\place, \bar{\Omega}]_{T}$
  and $[\place, \bar{\Omega}]_{\cat{D}}$. It is enough to show
  that, for a $T$-algebra $A_a$, an object $M \in \cat{D}$ and
  $f \colon M \to A^*$ in $\cat{D}$,
  the commutativity of the two diagrams in (\ref{diag:proof-state-pred3})
  are equivalent:
  \begin{align}
    \label{diag:proof-state-pred3}
    \xymatrix@R-1.5em{
      M \ar[r]^{f} \ar[d]_{f}
      & A^* \ar[d]^{\tau^{\sharp}} \\
      A^* \ar[r]_{a^*}
      & (TA)^*
    } \quad
    \xymatrix@R-1.5em{
      TA \ar[r]^-{Tf^{\sharp}} \ar[d]_{a}
      & T \p*{M^*} \ar[d]^{\zeta_M} \\
      A \ar[r]_-{f^{\sharp}}
      & M^*
    }
  \end{align}

  Since commutativity of the left diagram in (\ref{diag:proof-state-pred3})
  is equivalent to that of the left diagram in (\ref{diag:proof-state-pred}),
  it is enough to show the diagram in (\ref{diag:proof-state-pred2})
  commutes.
  \begin{align}
    \label{diag:proof-state-pred}
      \xymatrix@R-1.5em{
        TA \ar[r]^-{\tau} \ar[d]_{a}
        & A^{**} \ar[d]^{f^*} \\
        A \ar[r]_-{f^{\sharp}}
        & M^*
      }
      \xymatrix@R-1.5em{
        A^*
        & M \ar[l]_{f} \ar[ld]^{\eta} \\
        M^{**} \ar[u]^{(f^{\sharp})^*}
      }
  \end{align}
  \begin{align}
    \label{diag:proof-state-pred2}
    \vcenter{
      \xymatrix@R-1.5em{
        TA \ar[r]^-{\tau} \ar[d]_{Tf^{\sharp}}
        & A^{**} \ar[r]^{f^*} \ar[d]_{(f^{\sharp})^{**}}
        & M^* \\
        T \p*{M^*} \ar[r]_-{\tau}
        & M^{***} \ar[ur]_{\eta^*}
      }
    }
  \end{align}
  The left square commutes
  by the naturality, and the right triangle commutes since
  the right diagram in (\ref{diag:proof-state-pred})
  does by the adjointness.
\end{proof}

\subsection{Proof of Theorem~\ref{thm:partialHealthiness}}
\begin{proof}
  It is easy to check
  the functor $\bbP^{\tau}$ coincides with the composite
  \begin{equation}
    \Kl(T) \xxto{\Kl(\tau)} \Kl(\cat{D}(\OmegaD^{(\place)}, \OmegaD))
      \xxto{K} \cat{D}^{\op} \enspace ,
  \end{equation}
  where $K$ is the comparison functor and $\Kl(\tau)$
  denotes the functor defined by
  \begin{align*}
    \Kl(\tau) (X) &= X \enspace , \quad \text{and} \\
    \Kl(\tau) \bigl(\,X \xxto{f} Y \;\text{(in $\Kl(T)$)}\,\bigr) &=
      \bigl(\,X \xxto{f} TY \xxto{\tau} \cat{D}(\OmegaD^{Y}, \OmegaD)
 \;\text{(in $\Sets$)}
\,\bigr)\enspace.
  \end{align*}

  Since $K$ is full and faithful,
  it is enough to show the action
  $\Kl(\tau) \colon \Kl(T)(X, Y) \to \Kl(\cat{D}(\OmegaD^{(\place)}, \OmegaD))(X, Y)$
  is surjective (resp. injective).
  By the definition of $\Kl(\tau)$,
  this action is the postcomposition map by $\tau$.
  \begin{equation}
    \tau_{*} \colon \Set(X, TY) \to \Set(X, \cat{D}(\OmegaD^{Y}, \OmegaD))
  \end{equation}
  Here we use the identification
  $\Kl(T)(X, Y) = \Set(X, TY)$ and
  $\Kl(\cat{D}(\OmegaD^{(\place)}, \OmegaD))(X, Y) = \Set(X, \cat{D}(\OmegaD^{Y}, \OmegaD))$.
  When $\tau$ is injective, the postcomposition $\tau_{*}$
  is injective by the definition of mono.
  When $\tau$ is surjective, it is split epi (by the axiom of choice)
  hence so is $\tau_{*}$.
\end{proof}

\subsection{Proof of Lemma~\ref{lem:finitary-factorization}}
\begin{proof}
  Since $T$ is finitary, $t \colon 1 \to TX$ factor through
  some finite subset $s \colon X' \into X$ as $t = Ts \circ t'$.
  Then we get a desired factorization as follows:
  \begin{align*}
    \xymatrix@R-1.5em{
      & TX' \ar[d]^-{Ts} \\
      1 \ar[ur]^-{t'} \ar[r]_-{t}
      & TX
    }
    \xymatrix@R-1.5em{
      & \Omega^{TX'} \ar[dl]_-{(t')^*}
      & \Omega^{n} \ar[l]_-{\tau^{\sharp}_{X'}} \\
      \Omega
      & \Omega^{TX} \ar[l]^-{t^*} \ar[u]_-{(Ts)^*}
      & \Omega^{X} \ar[u]_-{s^*} \ar[l]^-{\tau^{\sharp}_X} \mathrlap{\enspace .}
    }
  \end{align*}
\vspace{-2em}
\end{proof}

\subsection{Proof of Theorem~\ref{thm:finitary-healthiness}}
\begin{proof}
  We can assume $X = 1$.

  \emph{Only if:} It follows from Corollary~\ref{cor:finitary-pred-transf}.

  \emph{If:} The statement is obviously true for $Y = \emptyset$,
  so we assume $Y \neq \emptyset$.
  Since $\phi$ is finitary, we can decompose $\phi$
  as $\phi = \phi' \after s^*$ for some finite subset $s \colon Y' \into Y$
  and $\phi' \colon \Omega^{Y'} \to \Omega$.
  We can assume $Y'$ is nonempty.
  Fix a retraction $r \colon Y \onto Y'$.
  We have $\phi' = \phi \after r^*$, then $\phi'$ also lifts
  to an $\cat{D}$-morphism $\bar{\phi'}$.
  By the surjectiveness of $\sigma_{Y'}$, there exists
  some $t' \in T{Y'}$ with $\sigma_{Y'}(t') = f'$.
  Take $t = Ts(t')$, then we have $\sigma_X(t) = \phi$,
  which concludes the proof.
\end{proof}

\subsection{Proof of Theorem~\ref{prop:finitary-continuous}}

We use the following lemma on elementary topology.
It easily follows from the compactness of $\Omega^X$.

\begin{lemma}
  For a finite discrete space $\Omega$ and an arbitrary set $X$,
  each clopen set $C \subset \Omega^{X}$ is written as
  $(s^{*})^{-1}(S)$ for some finite subset $s \colon X' \into X$
  and some subset $S$ of $\Omega^{X'}$.
\end{lemma}

\begin{proof}[Proof of Proposition~\ref{prop:finitary-continuous}]
  We can assume $X = 1$.

  \emph{If}: assume $\phi$ is continuous. For each $z \in \Omega$,
  the inverse image $\phi^{-1}(z)$ is clopen, so it can be described as
  $(\iota_z^*)^{-1}(S_z)$ for some
  finite subset $\iota_z \colon Y_z \into Y$ and $S_z \subset \Omega^{Y_z}$.
  Then $\iota \colon Y' = \bigcup_{z \in \Omega} Y_z \into Y$ is still finite
  and $\phi$ factors through $\iota^*$.

  \emph{Only if}: assume $\phi$ is finitary;
  we have $s \colon n \to Y$ and $\phi' \colon \Omega^n \to \Omega$
  such that $\phi' \circ s^* = \phi$.
  The map $s^*$ is obviously continuous, and
  so is $\phi'$ since its domain $\Omega^n$ is (finite) discrete.
  Therefore their composite $\phi$ is continuous.
\end{proof}


\end{document}